\documentclass[twoside]{article}

\usepackage[preprint]{aistats2027}
\usepackage[round]{natbib}

\usepackage[utf8]{inputenc} 
\usepackage{hyperref}       
\usepackage{url}            
\usepackage{booktabs}       
\usepackage{amsfonts}       
\usepackage{nicefrac}       
\usepackage{microtype}      
\usepackage{xcolor}         
\newcommand{\added}[1]{#1}
\newcommand{\deleted}[1]{}

\usepackage{multirow}
\usepackage{amsmath, amsthm, mathrsfs}
\usepackage{caption}
\usepackage{mathtools}
\usepackage{float}
\usepackage{adjustbox}
\usepackage{graphicx}
\usepackage{bm}
\usepackage{amssymb}
\usepackage{subcaption}
\usepackage{algorithm}
\usepackage{algorithmic}
\usepackage{comment}

\setcitestyle{authoryear,open={(},close={)}}

\def\be{\begin{equation}}
\def\ee{\end{equation}}
\def\bea{\begin{eqnarray}}
\def\eea{\end{eqnarray}}
\def\bd{\begin{displaymath}}
\def\ed{\end{displaymath}}
\def\bda{\begin{eqnarray*}}
\def\eda{\end{eqnarray*}}
\def\bsm{\begin{small}}
\def\esm{\end{small}}

\def\bZ{\mathbf{Z}}

\def\bg{\begin{figure}[tbph]\begin{center}}
\def\eg{\end{center}\end{figure}}

\def\bc{\begin{center}}
\def\ec{\end{center}}
\def\bremark{\begin{remark}}
\def\eremark{\end{remark}}

\def\bV{\mathbf{V}}

\def\bW{\mathbf{W}}

\def\1{\mathbf{1}}

\def\bH{\mathbf{H}}

\def\bS{\mathbf{S}}

\def\bX{\mathbf{X}}

\def\bZ{\mathbf{Z}}

\def\vec{\mbox{vec}}
\def\bg{\mathbf{g}}

\def\bA{{\bf A}}

\def\ha1{\hat \beta_1}

\def\bnt{\begin{enumerate}}
\def\ent{\end{enumerate}}

\def\bsc{\begin{scriptsize}}
\def\esc{\end{scriptsize}}

\newtheorem{remark}{Remark}
\newtheorem{tm}{Theorem}

\def\CK{{\cal K}}

\def\by{\mathbf{y}}

\def\KK1{\CK_{h_1,\vec\lambda_1}}

\def\bX{{\bf X}}

\def\ss2{\sum_{x^u\in {\cal X}^u}\sum_{y^u\in {\cal D}_{x^u}^1}  \beta_{\lambda_2}(x^u, y^u)}

\makeatletter
\newcommand{\figcaption}{\def\@captype{figure}\caption}
\newcommand{\tabcaption}{\def\@captype{table}\caption}
\makeatother

\begin{document}

\raggedbottom

\twocolumn[
\aistatstitle{Subdata Selection for Measurement-Constrained Linear Regression under Model Misspecification}

\aistatsauthor{Yundi Kong \And Min Yang \And Ping-Shou Zhong}

\aistatsaddress{
Department of Mathematics, Statistics, and Computer Science\\
University of Illinois at Chicago, Chicago, IL 60607\\
\texttt{\{ykong21, myang2, pszhong\}@uic.edu}
} ]
\begin{abstract}
In many modern applications, obtaining response variables is costly, time-consuming, or technically challenging, making it impractical to measure outcomes for all available units. In such settings, covariates may be available for a large pool of candidates, but responses can only be collected for a small, carefully selected subset, making the choice of which units to measure a consequential decision. Most existing subsampling methods are developed for massive datasets, emphasizing the trade-off between statistical accuracy and computational cost, and typically relying on the assumption of a correctly specified model. These features limit their suitability in measurement-constrained settings, where statistical efficiency is paramount and model misspecification is often unavoidable. To address these challenges, we develop a new subset selection framework for linear regression that accommodates potential model misspecification. The proposed approach targets predictive expected loss and yields an $A$-optimality–type criterion that depends only on covariate information, requiring no response observations at the selection stage. We refer to this method as Expected Loss Optimal Selection (ELOS). By focusing on predictive risk rather than model-based optimality, ELOS remains well-defined and robust under misspecification. Simulation studies and real data applications demonstrate that ELOS achieves strong predictive performance and compares favorably with existing methods across a range of scenarios.
\end{abstract}

\section{Introduction}
In many modern applications, large collections of candidate data points are available, but obtaining the corresponding response is costly, time-consuming, or technically difficult. For example, materials science experiments may involve many possible experimental settings, while only a limited number of synthesis trials can be conducted \citep{Nakamura2017}. In computer engineering, benchmark evaluations of processor performance are expensive \citep{Ein_Dor1987}, and environmental monitoring may permit measurements at only a limited number of locations due to sensor or infrastructure constraints \citep{ChenVarma2015}. These examples motivate measurement-constrained regression, where covariates are observed for a large candidate pool, but responses can be measured only for a small selected subset. The central challenge is therefore to determine which observations should be measured to achieve the best statistical efficiency. This setting is related to experimental design, but differs from classical design problems in which experimental conditions can often be specified freely. Here, the design points are fixed in advance, and one must select from a potentially massive pool of candidate covariate vectors. Moreover, the selection rule must be constructed before observing the responses: the full design matrix is available at the selection stage, responses are then collected only for the selected observations. This response-unavailability constraint distinguishes the problem from much of the existing subset selection and subsampling literature.

Most existing subsampling methods are developed for massive data analysis, where the full response vector is assumed to be available and subsampling is used mainly to balance the trade-off between statistical accuracy and computational cost. Their selection criteria are often derived from Fisher information, asymptotic variance, or related model-based design principles. Recent reviews summarize the rapid development of this literature \citep{YaoWang2021Review,YuAiYe2024Review}. For linear regression, \citet{Ma2015} studied algorithmic leveraging based on statistical leverage scores. For generalized linear models and classification problems, related methods include local case-control sampling for imbalanced logistic regression \citep{FithianHastie2014LCC} and optimal subsampling for large-sample logistic regression \citep{Wang2018}. These ideas have been extended to quantile regression \citep{WangMa2021Quantile,AiWangYuZhang2021Quantile}, general big data regression procedures \citep{AiYuZhangWang2021General}, distributed quasi-likelihood estimation \citep{YuWangAiZhang2022Distributed}, and feature screening for massive data analysis \citep{ZhuPanWuWang2022Screening}. Deterministic methods such as Information-Based Optimal Subdata Selection (IBOSS) \citep{wang2019information} select covariate extremes to maximize the determinant of the information matrix, providing a scalable design-based strategy for linear regression. Further developments include orthogonal subsampling for linear models \citep{WangElmstedtWongXu2021Orthogonal} and information-based optimal subdata selection for logistic regression \citep{ChengWangYang2020IBOSSLogistic}.   Although many of these methods use only covariates during selection and can therefore be adapted to measurement-constrained regression, relatively little work has explicitly studied subset selection under response measurement constraints. To the best of our knowledge, \cite{WangYuSingh2017} and \cite{Meng2021LowCon} are among the few studies that directly address this setting, proposing a subset selection strategy mainly for linear models, with possible extensions to generalized linear models. 


However, most of the methods discussed above are developed under correct model specification, which can be restrictive in practice. Real-world data may contain omitted variables, nonlinear structure, heteroscedasticity, or other departures from the working model. Under misspecification, criteria focused only on parameter estimation under the nominal model may no longer align with predictive accuracy. This motivates selection criteria that are robust to misspecification, use only covariates at the selection stage, and maintain strong prediction performance. LowCon \citep{Meng2021LowCon} addresses this issue from a design-based perspective by selecting data points that approximate a low-correlation design under possible linear-model misspecification.
Another natural direction is to move from parameter-estimation criteria toward loss-based, prediction-oriented objectives. Coreset methods follow this principle by approximating the empirical loss function of the full dataset rather than optimizing parameter estimation directly. For least-squares regression, \cite{Boutsidis2013} showed that a weighted subset can achieve a $(1\pm\varepsilon)$ relative-error approximation to the full least-squares objective. Extensions include $\ell_p$ regression \citep{Woodruff2024} and logistic regression \citep{Munteanu2018}. These results highlight the appeal of loss-based subset selection when prediction is the primary goal. However, standard coreset constructions typically require access to responses when forming the subset, limiting their applicability when responses are unavailable at the selection stage.

Motivated by these observations, we develop a prediction-oriented framework for measurement-constrained subdata selection. We derive the corresponding objective under correctly specified and misspecified linear models and show that minimizing the resulting expected loss leads to an A-optimality--type criterion. 
\added{The results build an explicit connection between classical prediction-oriented optimal design and measurement-constrained subdata selection, with Theorem~3, optimality of the classical prediction criterion under model misspecification, providing the new theoretical component that has no counterpart in the classical theory.}
This yields a principled alternative to Fisher-information-based designs without requiring response information at the selection stage. We refer to the proposed criterion as Expected Loss Optimal Selection (ELOS). Empirical results show that, compared with simple random sampling (SRS), leverage-score sampling, IBOSS, LowCon and coreset methods, ELOS achieves improved or comparable predictive performance with competitive computational efficiency across a range of settings. This paper is organized as follows. Section \ref{sec:2} introduces notation and reviews related work on subdata selection. Section \ref{Section Proposed} presents the main theoretical results. Section \ref{Num_Exp} evaluates the proposed method through simulation studies and real data examples. Section \ref{discussion} concludes with a  discussion and future work. Proofs and additional experiments are provided in the Appendix.

\section{Notations and related work in subdata selection}
\label{sec:2}

\subsection{Problem setup and notation}

Consider a linear regression model with covariate matrix $\bX \in \mathbb{R}^{n \times d}$
and response vector $\bm{y} \in \mathbb{R}^n$,
\[
\bm{y} = \bX\bm{\beta} + \bm{\varepsilon}, \ 
\mathbb{E}(\bm{\varepsilon})=\bm{0} \text{ and } V(\bm{\varepsilon})=\sigma^2 {\bm I}_n,
\]
where $\bm{x}_i^\top$ denotes the $i$th row of $\bX$ and $y_i$ is the corresponding response. The ordinary least squares (OLS) estimator minimizes the full-data squared loss
\[
\mathcal{L}(\bm{\beta}) = \| \bm{y} - \bX\bm{\beta} \|^2 
= \sum_{i=1}^n (y_i - \bm{x}_i^\top \bm{\beta})^2.
\]
Assuming $\bX$ has full column rank, the minimizer has the closed-form expression
\[
\hat{\bm{\beta}}_{\text{OLS}} 
= \arg\min_{\bm{\beta} \in \mathbb{R}^d} \mathcal{L}(\bm{\beta})
= (\bX^\top \bX)^{-1} \bX^\top \bm{y}.
\]
In the context of subdata selection, estimation is often performed using only a subset. Such procedures can be naturally formulated as a weighted least square (WLS) problem, in which selected observations contribute to estimation through assigned weights and each observation receives a nonnegative weight. 
Let the diagonal weight matrix $\bW = \operatorname{diag}(w_1, \ldots, w_n) \in \mathbb{R}^{n \times n}$. The weighted loss is
\[
\mathcal{L}_\bW(\bm{\beta}) 
= (\bm{y} - \bX\bm{\beta})^\top \bW (\bm{y} - \bX\bm{\beta}) 
= \sum_{i=1}^n w_i (y_i - \bm{x}_i^\top \bm{\beta})^2,
\]
with estimator
\(
\hat{\bm{\beta}}_{\text{WLS}} = (\bX^\top \bW \bX)^{-1} \bX^\top \bW \bm{y}.
\)

For later analysis, it is convenient to define the projection matrix
\(
H_W = \bX (\bX^\top \bW \bX)^{-1} \bX^\top \bW,
\)
so that the fitted values are $\hat{\bm{y}} = H_{\bW} \bm{y}$ and the residual sum of squares can be written as
\(
\mathcal{L}(\bW) = \| (I - H_{\bW}) \bm{y} \|^2.
\)
\subsection{Optimal sampling selection}
Optimal sampling selection seeks a small subset that yields statistically efficient estimation of $\bm{\beta}$ by optimizing criteria related to the variance-covariance matrix of the estimator. Existing approaches can be broadly categorized as probabilistic or deterministic.

Probabilistic subsampling methods assign each observation an inclusion probability \(\pi_i>0\) and draw a subset of size \(r\) accordingly. 
Let \(\delta_i\in\{0,1\}\) indicate whether the \(i\)th observation is included, and define
\(
\bW = \operatorname{diag}\!\left(\frac{\delta_i}{\pi_i}\right),
\)
the resulting weighted least squares estimator for linear regression is
\(
\hat{\bm{\beta}}_{\pi} = (\bX^\top \bW \bX)^{-1} \bX^\top \bW \bm{y}.
\)
The inclusion probabilities determine the statistical efficiency of \(\hat{\bm{\beta}}_{\pi}\), and are typically chosen to minimize an asymptotic mean squared error or related criterion under the working model \citep{Wang2018,AiYuZhangWang2021General}. 
For linear regression, a classical example is leverage-based sampling \citep{drineas2006sampling}, where
\(
\pi_i \propto h_{ii}, h_{ii}=\bm{x}_i^\top (\bX^\top \bX)^{-1}\bm{x}_i .
\)
Leverage scores measure the influence of observations on the geometry of the design matrix. Their statistical properties in linear regression, including bias and variance, were studied by \citet{Ma2015}. More generally, for logistic regression and generalized linear models, optimal probabilities can be derived from large-sample criterion-based approximations and need not be proportional to leverage scores alone \citep{Wang2018,AiYuZhangWang2021General}. 

Deterministic methods select observations without random sampling. Let \( {\bf{\Delta}} = \operatorname{diag}(\delta_1,\ldots,\delta_n) \), where \( \delta_i\in\{0,1\} \), denote the selection matrix. The corresponding estimator is
\[
\tilde{\bm{\beta}}_{\bf \Delta} = (\bX^\top {\bf\Delta} \bX)^{-1} \bX^\top {\bf\Delta} \bm{y}, 
\quad
\operatorname{Var}(\tilde{\bm{\beta}}_{\bf \Delta} )=\sigma^2(\bX^\top {\bf \Delta} \bX)^{-1}.
\]
Classical optimal design criteria include
\[
\begin{aligned}
\text{$D$-optimality: } 
&\max_{{\bf \Delta}} \det(\bX^\top {\bm \Delta} \bX), \\
\text{$A$-optimality: } 
&\min_{{\bf \Delta}} \operatorname{tr}\{(\bX^\top {\bf \Delta} \bX)^{-1}\},
\end{aligned}
\]
subject to $\sum_i \delta_i = r$. 
These criteria maximize Fisher information and improve the precision of \(\tilde{\bm{\beta}}_{\bf \Delta}\) under correct model specification.
The information-based optimal subdata selection (IBOSS) method \citep{wang2019information} is a representative deterministic approach. It constructs a subset using only the covariate matrix $X$ to approximate a D-optimal design, achieving near full-data efficiency under the linear model even when $r/n$ is small \citep{wang2019information}.

\subsection{Coreset method}

In contrast to model-based optimal subsampling, coreset methods adopt a geometric and loss-approximation perspective. 
For linear regression, a coreset is a small weighted subset designed to preserve the full-data least-squares loss. 
Let $\bS \in \mathbb{R}^{r \times n}$, with $r \ll n$, be a sampling--rescaling matrix that encodes both row selection and weight adjustment. 
The coreset loss is
\[
\mathcal{L}_\bS(\bm{\beta}) 
= \| \bS\bm{y} - \bS\bX\bm{\beta} \|_2^2
= (\bm{y} - \bX\bm{\beta})^\top \bS^\top \bS(\bm{y} - \bX\bm{\beta}),
\]
with estimator
\(
\hat{\bm{\beta}}_{\text{coreset}} 
= (\bX^\top \bS^\top \bS \bX)^{-1} \bX^\top \bS^\top \bS\bm{y}.
\)
To connect coreset methods with other subdata selection procedures, we write
\(
\bW = \bS^\top \bS = \operatorname{diag}(w_1,\ldots,w_n),
\)
where $w_i=0$ for unselected observations and $w_i>0$ for selected observations. 
Then the coreset loss can be equivalently expressed as
\[
\mathcal{L}_\bS(\bm{\beta})
=
(\bm{y}-\bX\bm{\beta})^\top \bW(\bm{y}- \bX\bm{\beta})
=
\sum_{i=1}^n w_i (y_i-\bm{x}_i^\top\bm{\beta})^2,
\]
and the estimator becomes
\[
\hat{\bm{\beta}}_{\text{coreset}}
=
(\bX^\top \bW \bX)^{-1}\bX^\top \bW\bm{y}.
\]
Thus, the coreset estimator is a weighted least squares estimator indexed by a sparse weight matrix $W$, a representation that will also be used to describe other subset-based estimators.

Coreset quality is typically measured by how well the weighted loss approximates the full-data loss. 
Let $\hat{\bm{\beta}}_{\text{OLS}}$ and $\hat{\bm{\beta}}_{\text{coreset}}$ denote the minimizers of the full-data and coreset losses, respectively. 
For least-squares regression, a coreset of size $O(d/\varepsilon^2)$ can guarantee \citep{Boutsidis2013}:
\[
\| \bm{y} - \bX\hat{\bm{\beta}}_{\text{OLS}} \|_2^2
\le
\| \bm{y} - \bX\hat{\bm{\beta}}_{\text{coreset}} \|_2^2
\le
(1+\varepsilon)
\| \bm{y} - \bX\hat{\bm{\beta}}_{\text{OLS}} \|_2^2
.
\]
This prediction-oriented guarantee makes coreset methods appealing under model misspecification, since it is formulated in terms of loss approximation rather than Fisher-information-based parameter efficiency and it doesn't rely on probabilistic assumptions about the data-generating process. 
Their performance depends on the sampling--rescaling strategy and the subdata size $r$, which together determine how well the weighted subset approximates the geometry of the full data \citep{Drineas_JMLR_2012,Boutsidis2013}. However, the strongest coreset guarantees and many practically effective coreset constructions for least-squares regression are response-aware, using $\bm y$ when forming the subset or weights. 
For example, the near-optimal deterministic constrained-regression construction of \citet{Boutsidis2013} uses the augmented matrix $[\bX,\bm y]$ and is therefore response-dependent.
Although \citet{Boutsidis2013} established the existence of a response-agnostic coreset construction and provided theoretical justification for its use in measurement-constrained settings, its
worst-case approximation ratio is $O(n/r)$, substantially weaker than
the near-optimal guarantee of the response-aware construction, and the
paper does not provide an implementation algorithm.

\subsection{Connection to optimal subsampling and coreset methods}
Optimal sampling selection and coreset methods offer complementary perspectives for subdata selection.
Under measurement constraints, the primary goal is to select informative observations before responses are measured; at the same time, scalability remains important when the candidate pool is large.
Coreset methods aim to approximate the least-squares loss to preserve predictive accuracy, whereas optimal subsampling methods focus on variance–covariance–based criteria to enhance parameter estimation efficiency. The former are inherently prediction-oriented and tend to be more robust to model misspecification, while the latter can achieve high statistical efficiency under correct model specification. Motivated by these complementary strengths, we evaluate subdata selection through the expected prediction loss of the estimator induced by a sparse weight matrix $W$.
Formally, we define
\[
\mathbb{E}[\mathcal{L}(\bW)] 
= 
\mathbb{E}\big[\| \bm{y} - H_\bW \bm{y} \|^2\big],
\]
where $H_\bW = \bX (\bX^\top \bW \bX)^{-1} \bX^\top \bW$ is the projection matrix associated with the weighted or subset-based estimator defined by $\bW$.
This criterion bridges the loss-approximation perspective of coreset methods and the statistical-efficiency perspective of optimal subsampling, while directly targeting predictive performance after response measurement. It captures both variance and bias components and provides a principled foundation for subdata selection under measurement constraints and potential model misspecification. The next section develops this framework in detail.
\section{Main results}
\label{Section Proposed}
This section develops the theoretical foundation of ELOS. We consider deterministic subset selection using only covariate information, prior to response measurement. Through a loss-based formulation, we derive the expected prediction error under both correct and misspecified linear models, and show that minimizing this loss leads to an A-optimality-type criterion for subset selection. \added{Throughout Theorems~\ref{thm1}--\ref{thm:3}, the design matrix is treated as fixed and the errors are homoscedastic.}

\subsection{Under correctly specified linear model}
\label{model_correct}
We first consider the correctly specified linear model
\begin{equation}
\label{eq:well_specified_model}
\bm{y} = \bX\bm{\beta} + \bm{\varepsilon}, 
\; \mathbb{E}(\bm{\varepsilon})=\bm{0} \text{ and } \operatorname{Var}(\bm{\varepsilon})=\sigma^2 {\bm I}_n,
\end{equation}
where \(\bX \in \mathbb{R}^{n \times d}\), \(\bm{y} \in \mathbb{R}^n\), and \(\bm{\beta} \in \mathbb{R}^d\) is the true coefficient vector.

Following the deterministic subset-selection setup, let
\[
{\bf \Delta} = \operatorname{diag}(\delta_1,\ldots,\delta_n), 
\qquad \delta_i \in \{0,1\}, 
\qquad \sum_{i=1}^n \delta_i = r,
\]
denote the selection matrix. The subdata-based estimator is
\begin{equation}
\label{eq:subdata_estimator_delta}
\tilde{\bm{\beta}}_{\bf \Delta}
=
(\bX^\top {\bf \Delta} \bX)^{-1}\bX^\top {\bf \Delta} \bm{y}.
\end{equation}
The associated projection matrix is
\begin{equation}
\label{eq:projection_delta}
\bH_{\bf \Delta}
=
\bX(\bX^\top {\bf \Delta} \bX)^{-1}\bX^\top {\bf \Delta},
\end{equation}
and the full-design prediction loss is
\begin{equation}
\label{eq:prediction_loss_delta}
\mathcal{L}({\bf \Delta})
=
\|\bm{y}-\bH_{\bf \Delta}\bm{y}\|^2
=
\|({\bf I}-\bH_{\bf \Delta})\bm{y}\|^2.
\end{equation}

The following result characterizes the expected loss under Model~\eqref{eq:well_specified_model} and shows that minimizing it is equivalent to an A-optimal subset selection criterion; see Appendix~\ref{app:proofs-main} for the proof.
\begin{tm}
\label{thm1}
Assume the linear model \( \bm{y} = \bX\bm{\beta} + \bm{\varepsilon} \), where $\mathbb{E}(\bm{\varepsilon})=\bm{0} \text{ and } \operatorname{Var}(\bm{\varepsilon})=\sigma^2 {\bm I}_n,$
and let the sampling matrix \( {\bf \Delta} = \operatorname{diag}(\delta_1, \ldots, \delta_n) \in \mathbb{R}^{n \times n} \) satisfy \( \delta_i \in \{0,1\} \), $\sum_{i=1}^n\delta_i=r$ and \( \operatorname{rank}({\bf \Delta} \bX) = d \). 
Then the expected loss is given by
\begin{equation}
\begin{aligned}
\mathbb{E}[\mathcal{L}({\bf \Delta})]
&= \sigma^2 \left[ n - 2d + \Psi_A({\bf \Delta}) \right].
\end{aligned}
\end{equation}
where \(
\Psi_A({\bf \Delta})
:= \operatorname{tr}\!\left[
(\bm{U}^\top {\bf \Delta} \bm{U})^{-1}
\right] \) with the compact singular value decomposition \( \bX = \bm{U} \bm{\Sigma} \bV^\top \). 
\end{tm}

\begin{remark}
The criterion
\(
\Psi_A({\bf \Delta})
=
\operatorname{tr}\{(\bm U^\top{\bf \Delta}\bm U)^{-1}\}
\)
has the classical $A$-optimal trace-inverse form in the left singular-vector representation. Equivalently,
\[
\Psi_A({\bf \Delta})
=
\operatorname{tr}\!\left[
\bX^\top\bX
(\bX^\top{\bf \Delta}\bX)^{-1}
\right],
\]
which is the finite-pool $V$-optimality criterion for prediction over the observed candidate pool. The distinction from classical exact design is that selection here is restricted to the fixed pool without replication, with
\(\delta_i\in\{0,1\}\) and \(\sum_i\delta_i=r\).
\end{remark}

\subsection{Under mis-specified linear model}
\label{model_mis}
We next consider the case where the data-generating process deviates from the assumed linear model:
\begin{equation}
\label{mis_model}
\bm{y} = \bX\bm{\beta} + f + \bm{\varepsilon},
\; \mathbb{E}(\bm{\varepsilon})=\bm{0} \text{ and } \operatorname{Var}(\bm{\varepsilon})=\sigma^2 {\bm I}_n,
\end{equation}
where \(f \in \mathbb{R}^n\) represents an unknown misspecification term. 
Since \(f\) is unavailable before selection, the subset is still chosen using only the covariates. 
Therefore, the estimator and projection matrix remain those in \eqref{eq:subdata_estimator_delta} and \eqref{eq:projection_delta}, respectively, with performance evaluated by the prediction loss in \eqref{eq:prediction_loss_delta}.

Let the compact singular value decomposition of \( \bX \in \mathbb{R}^{n \times d} \) be \( \bX = \bm{U} {\bf \Sigma} \bV^\top \), where \( \bm{U} \in \mathbb{R}^{n \times d} \). Extend \( \bm{U} \) to an orthonormal basis \( [\bm{U}\ \tilde{\bm{U}}] \in \mathbb{R}^{n \times n} \) such that \( \tilde{\bm{U}}^\top \bX = 0 \). While $f\in \mathbb{R}^n $ is unknown, because of the model format (\ref{mis_model}), we can always assume $f$ satisfy the following properties: 
\begin{itemize}
    \item[(i)] \( \bX^\top f = 0 \).
    \item[(ii)] \( \|f\|^2 / n = \tau^2 \text{ for some } \tau <\infty \).
    \item[(iii)] \( f = \tau \sqrt{n} \tilde{{\bm U}} \bm{a} \), where $\bm{a}$ is a direction vector of the unit sphere in \( \mathbb{R}^{n-d} \).
\end{itemize}
\begin{remark}
Any \( f \in \mathbb{R}^n \) can be decomposed as \( f = {\bm P_X} f + ({\bm I} - {\bm P_X})f \), where \( {\bm P_X}=\bX(\bX^\top \bX)^{-}\bX^\top \). The first term ${\bm P_X} f$ can be absorbed into the linear term, i.e., $\bX(\bm{\beta}+(\bX^\top \bX)^{-}\bX^\top f)$. Thus assumption (i) can be made without loss of generality.
\end{remark}
\begin{remark}
Assumption (ii) assumes the average squared magnitude of misspecification is bounded. This ensures \( f \) does not overwhelm the signal. It is practically reasonable. 
\end{remark}
\begin{remark}
Because of (i), \( f \in \operatorname{col}(\tilde{{\bm U}}) \) and it can be written as \( f = \tilde{{\bm U}} \bm{b} \) for some \( \bm{b} \in \mathbb{R}^{n-d} \). The boundedness condition \( \|f\|^2 / n = \tau^2 \) implies \( \|\bm{b}\| = \tau \sqrt{n} \). We represent \( \bm{b} \) as a scaled vector contained on the unit sphere in \( \mathbb{R}^{n-d} \) with fixed norm \( \tau \sqrt{n} \), i.e., 
    \(
    \bm{b} = \tau \sqrt{n} \bm{a}.
    \)

\end{remark}
The following theorem gives the expected loss under Model~\ref{mis_model}; see Appendix~\ref{app:proofs-main} for the proof.

\begin{tm}\label{thm:2}
Under Model \ref{mis_model}, suppose the sampling matrix \( {\bm \Delta} = \operatorname{diag}(\delta_1, \ldots, \delta_n) \) with \( \delta_i \in \{0,1\} \) satisfies $\sum_{i=1}^n\delta_i=r$ and \( \operatorname{rank}({\bm \Delta} \bX) = d \). 
Then the expected loss satisfies
\begin{equation}\label{eq:expected_loss}
\begin{aligned}
\mathbb{E}[\mathcal{L}({\bm \Delta})]
&= \sigma^2\left(n - 2d + \Psi_A({\bm \Delta})\right) \\
&\quad + \tau^2 n \cdot
\bm{a}^\top \tilde{\bm U}^\top
({\bm I} - \bH_{\bm \Delta})^\top 
({\bm I} - \bH_{\bm \Delta})
\tilde{\bm U}\bm{a}.
\end{aligned}
\end{equation}
where 
\(
\Psi_A({\bm \Delta}) := \operatorname{tr}\!\left[({\bm U}^\top {\bm \Delta} {\bm U})^{-1}\right].
\)
\end{tm}

While Theorem 2 provides an explicit expression for the expected loss, it depends on $\bm{a}$, which is unknown at the subset selection stage since no response data have been collected. Consequently, direct minimization of the expected loss is not feasible. To address this issue, we adopt the strategy of \cite{Adewale_Wiens_CSDA_2006} and instead minimize the average expected loss over all possible $\bm{a}$ on the unit sphere in \( \mathbb{R}^{n-d} \), that is, $E_{\bm{a}}\mathbb{E}[\mathcal{L}({\bf \Delta})]$, where $\bm{a}$ is uniformly distributed on the unit sphere. The following theorem shows that this average expected loss can be simplified to an expression involving $\Psi_A({\bm \Delta})$; see Appendix~\ref{app:proofs-main} for proof.

\begin{tm}\label{thm:3}
Under Model \ref{mis_model}, suppose the sampling matrix \( {\bm \Delta} = \operatorname{diag}(\delta_1, \ldots, \delta_n) \) with \( \delta_i \in \{0,1\} \) satisfies $\sum_{i=1}^n\delta_i=r$ and \( \operatorname{rank}({\bm \Delta} \bX) = d \). Assuming $\bm{a}$ is uniformly distributed on the unit sphere in  \( \mathbb{R}^{n-d} \), then 
\begin{equation}\label{ave_expected_loss}
E_{\bm{a}}\mathbb{E}[\mathcal{L}({\bm \Delta})]=\eta \cdot \left(n - 2d + \Psi_A({\bm \Delta})\right),
\end{equation}
where 
\[
\eta := \frac{\tau^2 n}{n - d} + \sigma^2,
\qquad 
\Psi_A({\bm \Delta}) := \operatorname{tr}\!\left[({\bm U}^\top {\bm \Delta} {\bm U})^{-1}\right].
\]
\end{tm}

\begin{remark}
Theorem~\ref{thm:3} is an average-case (Bayes-type) optimality statement under the uniform, rotation-invariant distribution on misspecification directions with fixed norm.
\end{remark}

Theorem~\ref{thm:3} shows that minimizing the average expected prediction loss leads to an A-optimality–type criterion that \added{is optimal for the averaged misspecification model specified above}. 
Unlike classical Fisher-information designs, this criterion directly targets prediction accuracy.
The derived criterion depends only on the design matrix \(X\) and not on the response vector \(\bm{y}\), enabling response-free implementation when responses are unavailable or delayed. 
Moreover, the parameters \(\sigma^2\) and \(\tau^2\) influence only the scaling of the objective, leaving the selection rule invariant.
\begin{algorithm}[H]
\caption{(ELOS implementation)}
\label{Algo1}
\begin{algorithmic}[1]
\STATE Use IBOSS \citep{wang2019information} to select an initial subdata set $S$ of size $r$ and let $\xi_S$ denote the corresponding approximate design.
\STATE Compute $F(\xi_S, {\bm u_i})$ for $i = 1, \ldots, n$, where ${\bm u_i}^{\top}$ is the $i$th row of ${\bm U}$, and $F(\xi_S, {\bm u_i})=\operatorname{tr}\{(\sum_{l\in S}{\bm u_l}{\bm u_l}^\top)^{-1}\}-{\bm u_i}^\top(\sum_{l\in S}{\bm u_l}{\bm u_l}^\top)^{-2}{\bm u_i}$ is the directional derivative at $\xi_S$ in the direction of the Dirac measure $\delta_{{\bm u_i}}$.
\STATE Form $E_1$ consisting of $\left\lceil r / {d} \right\rceil$ points ${\bm u_i} \notin S$ with the smallest values of $F(\xi_S, {\bm u_i})$.
\STATE Form $E_3$ consisting of $\left\lceil r / {d} \right\rceil$ points ${\bm u_i} \in S$ with the largest values of $F(\xi_S, {\bm u_i})$.
\STATE Update the subset: $S = (S \setminus E_3) \cup E_1$.
\STATE Repeat steps 2–5 for ${d}$ iterations.
\end{algorithmic}
\end{algorithm}

\subsection{ ELOS: expected loss optimal selection}
 The results from Theorems \ref{thm1} and \ref{thm:3} motivate our proposed ELOS method. The ELOS criterion aims to minimize the expected prediction loss, which reduces to minimizing an A-optimality–type criterion
\(
\Psi_A({\bm \Delta})
=
\operatorname{tr}\!\left[
({\bm U}^\top {\bm \Delta} {\bm U})^{-1}
\right],
\)
subject to
\(
\sum_{i=1}^n \delta_i=r.
\)
Here, $d$ denotes the number of columns of the compact left singular-vector matrix ${\bm U}\in\mathbb{R}^{n\times d}$.
\added{ELOS refers to the criterion $\Psi_A({\bm \Delta})$ itself, not to a particular numerical solver. IBOSS$+$ is used below as a scalable approximate optimizer. Appendix~\ref{app:optimizer_ablation} compares the same ELOS objective under a different first-order optimal-design optimizer.}
Identifying such an optimal subdata subset is a classical NP-hard problem due to the discrete nature of the selection variables.  
Building on optimal approximate design theory, \citet{Yang2025efficient} developed a new information-based methodology that yields subsets approaching the theoretical optimum.  
Their proposed algorithm IBOSS+ efficiently produces highly informative subdata with substantially reduced computational cost.  
 For practical implementation of ELOS in Algorithm~\ref{Algo1}, we employ the IBOSS$+$ algorithm \citep{Yang2025efficient} to select subsets that approximately minimize the expected prediction loss derived in Section~\ref{Section Proposed}.
 
\begin{figure*}[!t]
    \centering
    \includegraphics[
    width=0.86\textwidth,
    height=0.35\textheight
]{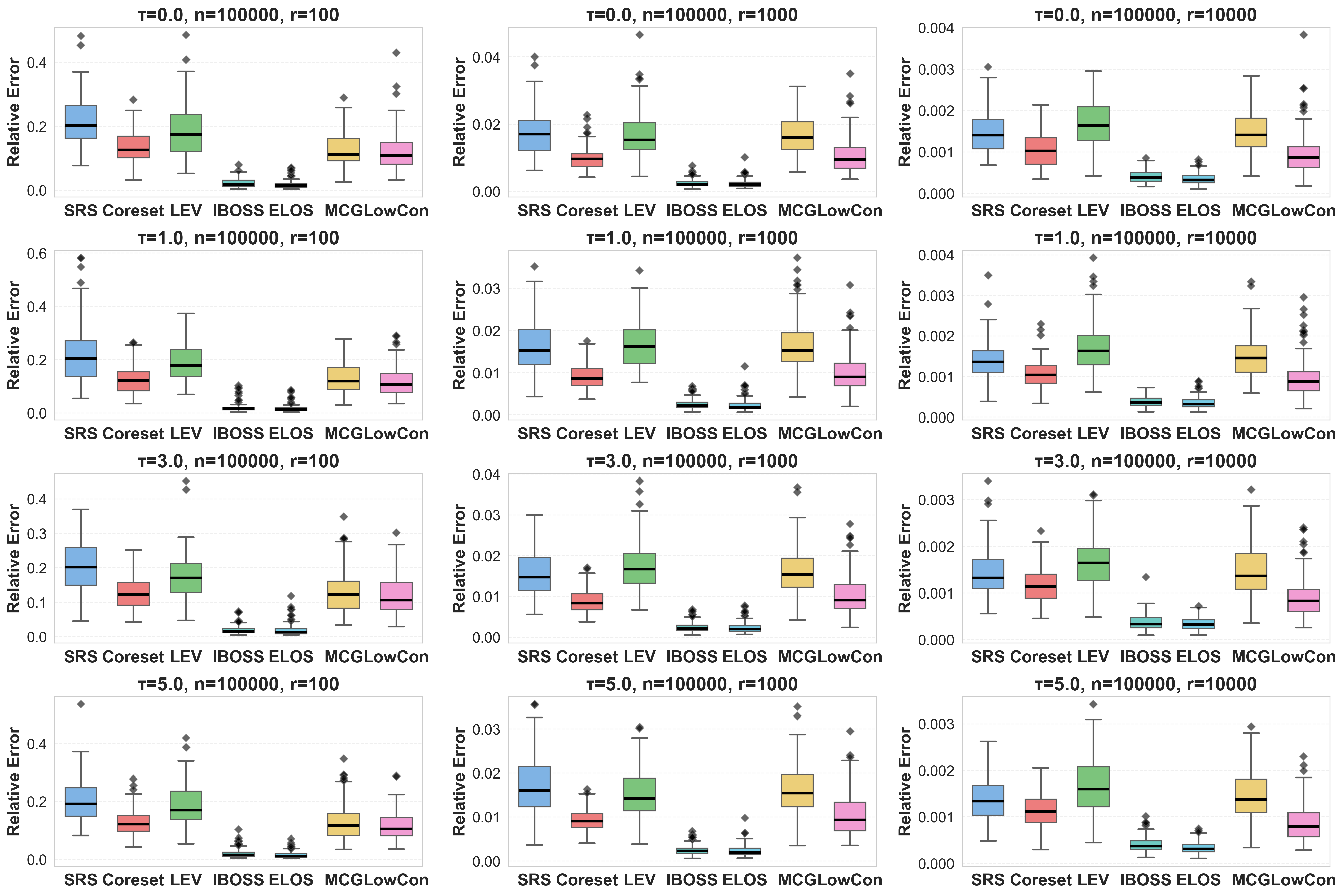}
    \caption{
    Empirical distributions of RE\(_{\mathrm{diff}}\) across 100 replications under \(Z_i\overset{\mathrm{i.i.d.}}{\sim}t_b(5,0,\Sigma)\).  
    }
    \label{fig:t_boxplot_n1e5}
\end{figure*}

\section{Numerical experiments}
\label{Num_Exp}
In this section, we investigate the empirical performance of the proposed subset selection framework through a series of simulation studies and real data examples. The goal is to evaluate how well ELOS performs in terms of prediction accuracy and computational efficiency when only a limited number of responses can be observed.
We compare the following methods: SRS, Coreset \citep{Boutsidis2013}, Leverage \citep{Ma2015}, IBOSS \citep{wang2019information}, MCG \citep{WangYuSingh2017}, LowCon \citep{Meng2021LowCon} and ELOS. 
Coreset is response-dependent and is included as a reference baseline for its robustness to misspecification, not as a measurement-constrained method. \added{Complete implementation and reproducibility details are provided in Appendix~\ref{app:implementation_details}.} 
The full-data OLS estimator
\(
\widehat{\bm\beta}_{\mathrm{full}}
\)
is used as the benchmark.
Predictive performance is measured by the relative error difference
\[
\mathrm{RE}_{\mathrm{diff}}
=
\frac{\|\hat{\bm y}_{\mathrm{sub}}-\bm y\|_2^2-\|\hat{\bm y}_{\mathrm{full}}-\bm y\|_2^2}
{\|\hat{\bm y}_{\mathrm{full}}-\bm y\|_2^2},
\]
where \(\hat{\bm y}_{\mathrm{sub}}\) and \(\hat{\bm y}_{\mathrm{full}}\) are the fitted responses from the subdata estimator and full-data OLS, respectively. Selection runtime includes method-specific preprocessing, such as the SVD for ELOS and the augmented-matrix SVD for coreset. Experiments were run on an Apple MacBook Pro (M4 Pro, 14-core CPU, 48 GB memory) using Julia 1.12.0 with single-threaded BLAS on CPU. 


\subsection{Simulation study}
\label{simu} \vspace{-5pt}
We generate a design matrix \(\bX=[\mathbf{1}_n\;\; \bZ]\in\mathbb R^{n\times {d}}\) including an intercept column and the remaining \(b=15\) columns for covariates, so that \({ d}=16\). For each replication, we draw \(\beta^\star\sim\mathcal N(0,{\bm I}_{d})\) and generate \(\by=\bX\beta^\star+\boldsymbol{f}(\bX)+\bm{\varepsilon}\), with \(\bm{\varepsilon}\sim\mathcal N(\mu\mathbf{1}_n,\sigma^2{\bm I}_n)\), \(\mu=1\), \(\sigma=1\), and $\boldsymbol{f}(\bX)=(f(\bm{x}_1),\ldots,f(\bm{x}_n))^T$. We vary \(\tau\in\{0,1,3,5\}\) for the misspecification term \(f(\bm{x}_i)=\tau\sin(\sum_{j=1}^{b}x_{ij}^2)\). For each subdata size \(r\), we report the relative error difference and the average selection runtime over 100 replications.
Two Scenarios are considered: (i) \(Z_i\overset{\mathrm{i.i.d.}}{\sim}t_b(5,0,{\bm \Sigma})\) and (ii) \(Z_i\overset{\mathrm{i.i.d.}}{\sim}\mathcal N_b(0,{\bm \Sigma})\), where \({\bm \Sigma}=5{\bm I}_b+5{\bm J_b}\) where \({\bm J}_b\) is the \(b\times b\) all-ones matrix. Figure~\ref{fig:t_boxplot_n1e5} presents boxplots of RE\(_{\mathrm{diff}}\) for Scenario (i), with the corresponding averaged RE\(_{\mathrm{diff}}\) and selection runtimes are reported in Table~\ref{tab:re_time_t5_pretest_withcoreset_sinsumsq_allcov_b15_sd1_rep100} (Appendix~\ref{simu:t}). Scenario (ii), which uses multivariate normal covariates, is deferred to Appendix~\ref{app:simulation-boxplots}.

Figure~\ref{fig:t_boxplot_n1e5} shows boxplots of RE\(_{\mathrm{diff}}\) under the heavytailed design across subdata sizes ($r \in \{100, 1000, 10000\}$). IBOSS and ELOS consistently achieve the lowest RE\(_{\mathrm{diff}}\), with ELOS performing comparably to or better than IBOSS in all settings. Both methods also exhibit substantially lower variability than the remaining competitors, particularly at small $r$. 
\added{Their similar performance reflects both favoring informative design points despite different objectives: IBOSS is $D$-optimality-oriented, whereas ELOS targets expected prediction loss. The optimizer ablation in Appendix~\ref{app:optimizer_ablation} further shows that this pattern is not driven by the IBOSS$+$ solver or its IBOSS initialization.} SRS, LEV, and MCG show considerably higher RE\(_{\mathrm{diff}}\), especially at $r=100$, while Coreset and LowCon improve over SRS and LEV but remain above IBOSS and ELOS throughout. 
\added{Coreset is a strong prediction benchmark because it uses the response through $[\bX,\bm y]$, an information advantage unavailable in the measurement-constrained setting.}
The relative ordering of methods changes little with the misspecification level $\tau$, \added{indicating stable empirical performance across the magnitudes examined here}. As $r$ increases, RE\(_{\mathrm{diff}}\) decreases and the gaps among methods narrow as larger subsets recover more information from the full data. Additional results are reported in Appendix~\ref{simu:t}, \added{computational and memory complexity analyses, including comparisons with competing methods, are provided in Appendix~\ref{app:elos_complexity_G}.}

\subsection{Real-data experiments}
\label{real_data}
We further evaluate the methods on two real datasets where the linear model is used as a potentially misspecified working approximation. To assess stability under a fixed design, we use residual bootstrap replicates. After fitting full-data OLS, we center and resample the residuals and form pseudo-responses \(\bm y^*=\hat{\bm y}+\tilde{\bm e}^*\). Each method is then rerun on \((\bX,\bm y^*)\), and performance is summarized over \(B=100\) bootstrap replicates. \added{Note that residual resampling does not preserve residual-covariate dependence and may therefore underrepresent some forms of structured misspecification.}
Additional details and references for the residual bootstrap are provided in Appendix~\ref{app:bootstrap}. 

\begin{figure}[!htpb]
    \centering    \includegraphics[width=0.65\columnwidth]{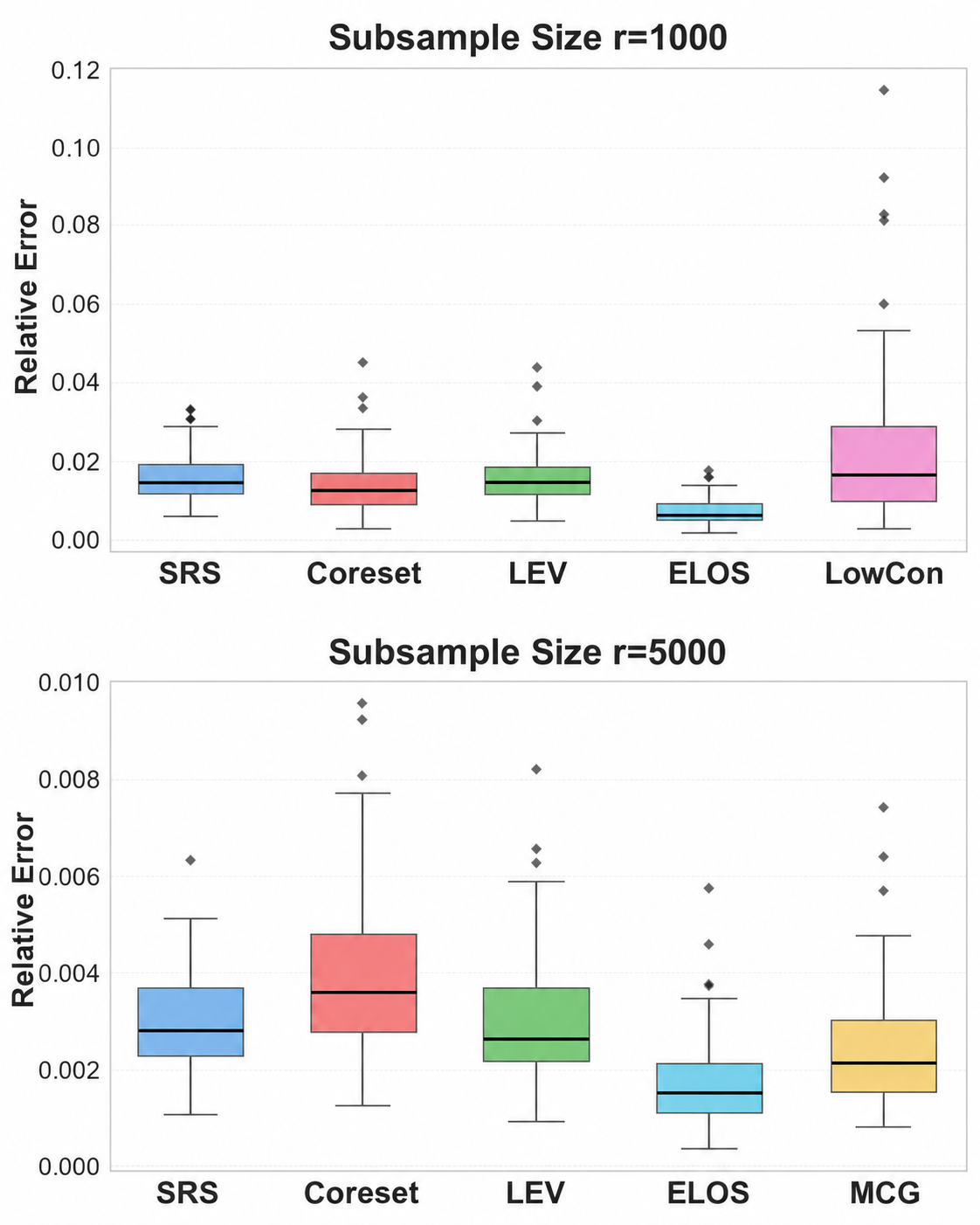}
    \caption{Chemical sensors data. Boxplots of RE\(_{\mathrm{diff}}\) across 100 bootstrap replicates for each subdata size \(r\). MCG at \(r=1000\) and LowCon at \(r=5000\) are omitted due to computational cost, and IBOSS is omitted for scale. Full results and runtimes are reported in Table~\ref{tab:sensor_realdata_residboot} (Appendix~\ref{test_train_sensor}). }
    \label{fig:boxplots_sensor}
\end{figure}

\subsubsection{Chemical sensors data}
\label{subsec:sensor_residboot}
We first consider the chemical sensors dataset of \cite{fonollosa2015dataset}. Following \cite{wang2019information}, we use the same preprocessing protocol: discarding the first 20{,}000 observations, taking the last sensor reading as the response, removing the second sensor due to negative readings, and applying a log transformation to the remaining variables. After preprocessing, the design has \(n=4{,}188{,}261\) observations and \(p=14\) covariates plus an intercept. Additional dataset details and train/test results are provided in Appendix~\ref{test_train_sensor}. Since the sensor-response relationship may involve nonlinear dynamics and interactions, the linear model is treated as a working approximation. For \(r\in\{1000,5000\}\), Figure~\ref{fig:boxplots_sensor} shows that ELOS achieves the smallest RE for both measurement budgets. In runtime, ELOS is slower than SRS, Leverage, and IBOSS, but is substantially faster than Coreset, LowCon and MCG and provides the best predictive accuracy (See Table~\ref{tab:sensor_realdata_residboot} in Appendix~\ref{test_train_sensor} for computation time).

\subsubsection{Million song release year prediction}
\label{subsec:year_msd_residboot}
We next evaluate the methods on YearPredictionMSD\citep{year_prediction_msd_203}, with \(n=515{,}345\), \(p=90\), and release year as the response. Since nonlinear models are often used for this task, the linear model is treated as a misspecified working approximation. We report results over \(B=100\) residual bootstrap replicates and \(r\in\{1000,10{,}000\}\). MCG is omitted for \(r=1000\) and LowCon is omitted for \(r=10000\) due to excessive computational cost. \added{Additional dataset preprocessing details and train/test results are reported in Appendix~\ref{test_train_song}. In the fully response-free setting using all \(k=90\) predictors, ELOS attains the smallest RE under both measurement budgets, as illustrated in Figure~\ref{fig:boxplots_yearmsd_k90}. To further examine how the competing methods behave as the predictor dimension changes, we also consider reduced designs with \(k=30\) and \(k=60\), constructed using Pearson-correlation screening. These additional experiments are reported in Appendix~\ref{test_train_song} and are intended as supplementary comparisons of statistical performance and computational scalability across different dimensions.}

\begin{figure}[htbp]
    \centering
    \includegraphics[width=0.7\columnwidth]{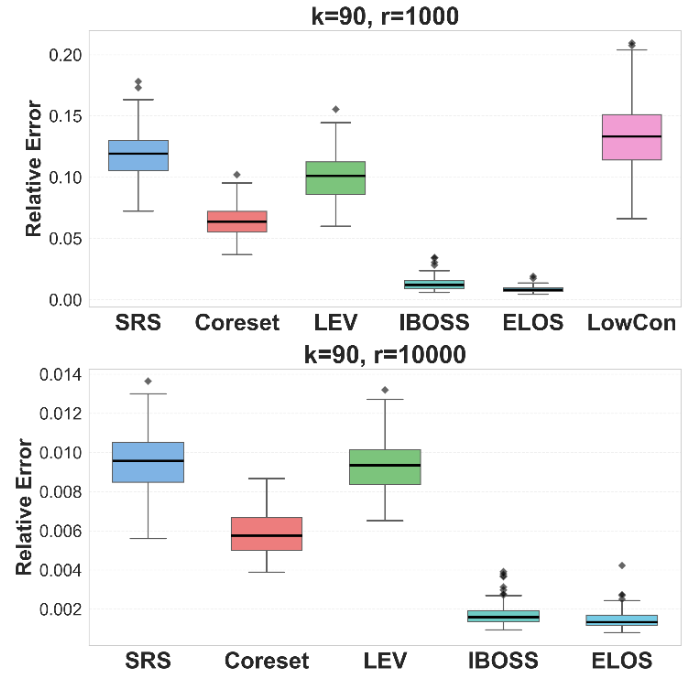}
    \caption{
    Million song release year prediction. Boxplots of RE\(_{\mathrm{diff}}\) across 100 bootstrap replicates for the YearPredictionMSD data using all \(k=90\) predictors. The two panels correspond to measurement budgets \(r=1{,}000\) and \(r=10{,}000\).
    }
    \label{fig:boxplots_yearmsd_k90}
\end{figure}

\section{Discussion}
\label{discussion}

This work develops ELOS for subset selection under measurement constraints and potential model misspecification, addressing the central question of how to identify the most informative subset using only covariate information when responses are costly to collect. ELOS uses a prediction-oriented expected-loss criterion that aligns selection more directly with predictive accuracy than parameter-estimation-based criteria. Compared with coreset methods, which require responses and serve as a \added{prediction-oriented reference benchmark}, ELOS \added{remains competitive across many settings while requiring no responses at the selection stage}. Among response-free baselines (SRS, leverage sampling, IBOSS, MCG, and LowCon), ELOS \added{generally achieves the best or near-best predictive accuracy while remaining scalable to large candidate pools}, offering a favorable balance among predictive performance, measurement feasibility, and computational scalability.
\added{This makes ELOS applicable to resource-constrained settings such as clinical trials, environmental monitoring, precision agriculture, and survey design.}

Several limitations suggest directions for future work.
First, under model misspecification, the expected loss depends on the unknown misspecification direction \(\bm a\), which is unavailable at the subset-selection stage,
\added{we therefore average over \(\bm a\) uniformly on the unit sphere. Theorem~\ref{thm:3} shows that the resulting loss reduces to \(\Psi_A(\bm\Delta)\), yielding a response-free criterion. This is an average-case guarantee rather than optimality for every fixed or structured departure, as also illustrated by the diagnostic experiments in Appendix~\ref{app:misspecification_diagnostic}.}
Alternative robustness criteria, such as minimax expected loss or distributionally robust objectives, may provide stronger protection against unfavorable misspecification directions.

Second, although the framework accommodates a broad class of model misspecifications, its performance may depend on how the deviation from the assumed linear structure is distributed. In situations where such deviations exhibit strong asymmetry, incorporating additional information could further enhance subset selection. For instance, a small pilot study with limited response observations may help identify a more appropriate working model and guide subsequent selection. This naturally leads to an interesting direction for future research: developing adaptive subset selection procedures that effectively integrate limited response information under measurement constraints.

\added{Third, the theory assumes a fixed design with homoscedastic errors; under heteroscedasticity, the expected-loss criterion requires a new derivation. The residual-bootstrap evaluation also does not preserve residual--covariate dependence, so alternative evaluation schemes may better capture structured misspecification.}

Finally, extending the expected-loss framework beyond linear regression, including generalized linear models and nonparametric regression, is an important direction for future work.
The exact response-free reduction obtained here relies on the linear least-squares projection structure. For nonlinear-in-parameters models, analogous local criteria would generally depend on unknown parameter values or pilot estimates and would introduce approximation error, so such extensions require new theoretical development rather than a direct application of the present results.
\newpage

\subsection*{AI Use Statement}

We did not use generative AI tools for any of the required-disclosure tasks listed in the AISTATS 2027 AI Policy.
We used generative AI tools for recommended-disclosure tasks including editing the manuscript to improve readability and grammar, LaTeX formatting, and debugging and editing existing research code.

All AI-assisted work was reviewed by the authors. Bibliographic information was checked against the original sources, and AI-assisted code changes were tested before use. The authors take responsibility for the final content of this work, including text, code, and other artifacts produced with the aid of generative AI.

\bibliographystyle{plainnat}
\bibliography{reference}

\clearpage
\appendix
\thispagestyle{empty}


\onecolumn
\section{Proofs for Section~\ref{Section Proposed}}
\label{app:proofs-main}

\subsection*{Proof of Theorem~\ref{thm1}}
\begin{proof}
Recall the loss function is defined as \( \mathcal{L}(\bm{\Delta}) = \| (\bm{I} - \bm{H}_{\bm{\Delta}})\bm{y} \|^2 \), where \( \bm{H}_{\bm{\Delta}} = \bm{X}(\bm{X}^\top \bm{\Delta} \bm{X})^{-1} \bm{X}^\top \bm{\Delta} \) is the projection matrix based on the subdata. 

Substitute \( \bm{y} = \bm{X}\bm{\beta} + \bm{\varepsilon} \) into the loss:
\begin{align*}
\mathcal{L}(\bm{\Delta}) &= \| (\bm{I} - \bm{H}_{\bm{\Delta}})(\bm{X}\bm{\beta} + \bm{\varepsilon}) \|^2 \\
&= \| (\bm{I} - \bm{H}_{\bm{\Delta}})\bm{X}\bm{\beta} \|^2 + \| (\bm{I} - \bm{H}_{\bm{\Delta}})\bm{\varepsilon} \|^2 + 2 \bm{\varepsilon}^\top (\bm{I} - \bm{H}_{\bm{\Delta}})\bm{X}\bm{\beta}.
\end{align*}

Take expectations,
\begin{align*}
\mathbb{E}[\mathcal{L}(\bm{\Delta})] 
&= \mathbb{E}\left[ \| (\bm{I} - \bm{H}_{\bm{\Delta}})(\bm{X}\bm{\beta} + \bm{\varepsilon}) \|^2 \right] \\
&= \| (\bm{I} - \bm{H}_{\bm{\Delta}})\bm{X}\bm{\beta} \|^2 + \mathbb{E}\left[ \| (\bm{I} - \bm{H}_{\bm{\Delta}})\bm{\varepsilon} \|^2 \right] + 2\, \mathbb{E}\left[ \bm{\varepsilon}^\top (\bm{I} - \bm{H}_{\bm{\Delta}})^\top (\bm{I} - \bm{H}_{\bm{\Delta}})\bm{X}\bm{\beta} \right] \\
&= \| (\bm{I} - \bm{H}_{\bm{\Delta}})\bm{X}\bm{\beta} \|^2 + \operatorname{tr} \left[ (\bm{I} - \bm{H}_{\bm{\Delta}})^\top (\bm{I} - \bm{H}_{\bm{\Delta}}) \cdot \sigma^2 \bm{I}_n \right] + 0 \\
&= \underbrace{\| (\bm{I} - \bm{H}_{\bm{\Delta}})\bm{X}\bm{\beta} \|^2}_{\text{Bias}^2} + \underbrace{\sigma^2 \operatorname{tr} \left[ (\bm{I} - \bm{H}_{\bm{\Delta}})^\top (\bm{I} - \bm{H}_{\bm{\Delta}}) \right]}_{\text{Variance}},
\end{align*}
the cross term vanishes because \( \mathbb{E}[\bm{\varepsilon}] = 0 \)
and the bias term \( \| (\bm{I} - \bm{H}_{\bm{\Delta}})\bm{X}\bm{\beta} \|^2 = 0 \). Then the expected loss becomes:
\[
\mathbb{E}[\mathcal{L}(\bm{\Delta})] = \sigma^2 \operatorname{tr} \left[ (\bm{I} - \bm{H}_{\bm{\Delta}})^\top (\bm{I} - \bm{H}_{\bm{\Delta}}) \right].
\]
Now expand the variance term:
\begin{align*}
(\bm{I} - \bm{H}_{\bm{\Delta}})^\top (\bm{I} - \bm{H}_{\bm{\Delta}}) &= \bm{I} - \bm{H}_{\bm{\Delta}} - \bm{H}_{\bm{\Delta}}^\top + \bm{H}_{\bm{\Delta}}^\top \bm{H}_{\bm{\Delta}}, \\
\Rightarrow \operatorname{tr} \left[ (\bm{I} - \bm{H}_{\bm{\Delta}})^\top (\bm{I} - \bm{H}_{\bm{\Delta}}) \right] &= n - 2 \operatorname{tr}(\bm{H}_{\bm{\Delta}}) + \operatorname{tr}(\bm{H}_{\bm{\Delta}}^\top \bm{H}_{\bm{\Delta}}).
\end{align*}
Note that \( \bm{\Delta} = \operatorname{diag}(\delta_1, \ldots, \delta_n) \) with \( \delta_i \in \{0, 1\} \) so \( \bm{\Delta} = \bm{\Delta}^2 \). We also assume \( \operatorname{rank}(\bm{X}) = \operatorname{rank}(\bm{\Delta} \bm{X}) = d \), ensuring that \( \bm{X}^\top \bm{\Delta} \bm{X} \) is invertible. 
Then we have:
\begin{align*}
\operatorname{tr}(\bm{H}_{\bm{\Delta}}) &= \operatorname{tr}\!\big(\bm{X}(\bm{X}^\top \bm{\Delta} \bm{X})^{-1} \bm{X}^\top \bm{\Delta}\big) = \operatorname{tr}\!\big((\bm{X}^\top \bm{\Delta} \bm{X})^{-1} \bm{X}^\top \bm{\Delta} \bm{X}\big) = \operatorname{tr}(\bm{I}_d) = d.
\end{align*}

Also,
\begin{align*}
\bm{H}_{\bm{\Delta}}^\top \bm{H}_{\bm{\Delta}} &= \bm{\Delta} \bm{X} (\bm{X}^\top \bm{\Delta} \bm{X})^{-1} \bm{X}^\top \bm{X} (\bm{X}^\top \bm{\Delta} \bm{X})^{-1} \bm{X}^\top \bm{\Delta}, \\
\Rightarrow \operatorname{tr}(\bm{H}_{\bm{\Delta}}^\top \bm{H}_{\bm{\Delta}}) &= \operatorname{tr}\!\left[(\bm{X}^\top \bm{\Delta} \bm{X})^{-1} \bm{X}^\top \bm{X} (\bm{X}^\top \bm{\Delta} \bm{X})^{-1} \bm{X}^\top \bm{\Delta}\, \bm{\Delta} \bm{X}\right] \\
&= \operatorname{tr}\!\left[(\bm{X}^\top \bm{\Delta} \bm{X})^{-1} \bm{X}^\top \bm{X} (\bm{X}^\top \bm{\Delta} \bm{X})^{-1} \bm{X}^\top \bm{\Delta} \bm{X}\right] \\
&= \operatorname{tr}\!\left[(\bm{X}^\top \bm{\Delta} \bm{X})^{-1} \bm{X}^\top \bm{X}\right].
\end{align*}

Let \( \bm{X} = \bm{U} \bm{\Sigma} \bm{V}^\top \) be the compact SVD of \( \bm{X} \), where \( \bm{U} \in \mathbb{R}^{n \times d} \), \( \bm{\Sigma} \in \mathbb{R}^{d \times d} \), and \( \bm{V} \in \mathbb{R}^{d \times d} \) with \( \bm{U}^\top \bm{U} = \bm{I}_d \), \( \bm{V}^\top \bm{V} = \bm{I}_d \). Then:
\[
\bm{X}^\top \bm{X} = \bm{V} \bm{\Sigma}^2 \bm{V}^\top, \quad \bm{X}^\top \bm{\Delta} \bm{X} = \bm{V} \bm{\Sigma} \bm{U}^\top \bm{\Delta} \bm{U} \bm{\Sigma} \bm{V}^\top.
\]
Therefore, the inverse of \( \bm{X}^\top \bm{\Delta} \bm{X} \) is:
\[
(\bm{X}^\top \bm{\Delta} \bm{X})^{-1} = \bm{V} \bm{\Sigma}^{-1} (\bm{U}^\top \bm{\Delta} \bm{U})^{-1} \bm{\Sigma}^{-1} \bm{V}^\top.
\]
Multiplying with \( \bm{X}^\top \bm{X} \), we get:
\[
(\bm{X}^\top \bm{\Delta} \bm{X})^{-1} \bm{X}^\top \bm{X} 
= \bm{V} \bm{\Sigma}^{-1} (\bm{U}^\top \bm{\Delta} \bm{U})^{-1} \bm{\Sigma}^{-1} \bm{V}^\top \cdot \bm{V} \bm{\Sigma}^2 \bm{V}^\top 
= \bm{V} \bm{\Sigma}^{-1} (\bm{U}^\top \bm{\Delta} \bm{U})^{-1} \bm{\Sigma} \bm{V}^\top.
\]
Taking the trace:
\[
\operatorname{tr}\!\left( (\bm{X}^\top \bm{\Delta} \bm{X})^{-1} \bm{X}^\top \bm{X} \right) 
= \operatorname{tr}\!\left( \bm{V} \bm{\Sigma}^{-1} (\bm{U}^\top \bm{\Delta} \bm{U})^{-1} \bm{\Sigma} \bm{V}^\top \right)=
\operatorname{tr}\!\left[(\bm{U}^\top \bm{\Delta} \bm{U})^{-1}\right].\] 
Therefore,
\begin{equation*}
\mathbb{E}[\mathcal{L}(\bm{\Delta})] = \sigma^2 \left[n - 2d + \Psi_A(\bm{\Delta}) \right],
\quad \text{with } \Psi_A(\bm{\Delta}) := \operatorname{tr}\!\left[(\bm{U}^\top \bm{\Delta} \bm{U})^{-1}\right]. \tag*{\qed}
\end{equation*}
\end{proof}

\subsection*{Proof of Theorem~\ref{thm:2}}
\begin{proof}
We begin with
\[
\mathcal{L}(\bm{\Delta}) = \| (\bm{I} - \bm{H}_{\bm{\Delta}})(\bm{X}\bm{\beta} + f + \bm{\varepsilon}) \|^2.
\]
Taking expectation over the noise \( \bm{\varepsilon} \sim \mathcal{N}(0, \sigma^2 \bm{I}_n) \), we obtain
\begin{align*}
\mathbb{E}_{\varepsilon}[\mathcal{L}(\bm{\Delta})] 
&= \mathbb{E}_{\varepsilon}\left[\| (\bm{I} - \bm{H}_{\bm{\Delta}})(\bm{X}\bm{\beta} + f + \bm{\varepsilon}) \|^2\right] \\
&= \left\| (\bm{I} - \bm{H}_{\bm{\Delta}})(\bm{X}\bm{\beta} + f) \right\|^2 + \sigma^2 \operatorname{tr}[(\bm{I} - \bm{H}_{\bm{\Delta}})^\top(\bm{I} - \bm{H}_{\bm{\Delta}})]\\
&=
 \| (\bm{I} - \bm{H}_{\bm{\Delta}})f \|^2 + \sigma^2 \operatorname{tr}[(\bm{I} - \bm{H}_{\bm{\Delta}})^\top(\bm{I} - \bm{H}_{\bm{\Delta}})].
\end{align*}

Now, using the assumption \( f = \tau \sqrt{n} \tilde{\bm{U}} \bm{a} \), we have 
\[
\| (\bm{I} - \bm{H}_{\bm{\Delta}})f \|^2 = \tau^2 n \cdot \bm{a}^\top \tilde{\bm{U}}^\top (\bm{I} - \bm{H}_{\bm{\Delta}})^\top (\bm{I} - \bm{H}_{\bm{\Delta}}) \tilde{\bm{U}} \bm{a}.
\]
On the other hand, in the proof of Theorem 1, it has been shown that
$\operatorname{tr}[(\bm{I} - \bm{H}_{\bm{\Delta}})^\top(\bm{I} - \bm{H}_{\bm{\Delta}})]= n - 2d + \Psi_A(\bm{\Delta})$. Thus, the conclusion follows. 
\end{proof}

\subsection*{Proof of Theorem~\ref{thm:3}}
\begin{proof}
Based on the expression in Theorem~\ref{thm:2}, to establish Theorem~\ref{thm:3} it suffices to take expectation over
\( \bm a \sim \mathrm{Unif}(\mathbb S^{n-d-1}) \) for the quadratic term
\[
\bm a^\top \tilde{\bm U}^\top (\bm{I}-\bm{H}_{\bm{\Delta}})^\top (\bm{I}-\bm{H}_{\bm{\Delta}})\tilde{\bm U} \bm a.
\]
Using the identity \(x^\top \bm{A} x=\operatorname{tr}(\bm{A}xx^\top)\), we have
\[
E_{\bm a}\!\left[\bm a^\top \tilde{\bm U}^\top (\bm{I}-\bm{H}_{\bm{\Delta}})^\top (\bm{I}-\bm{H}_{\bm{\Delta}})\tilde{\bm U} \bm a\right]
=
\operatorname{tr}\!\left(
\tilde{\bm U}^\top (\bm{I}-\bm{H}_{\bm{\Delta}})^\top (\bm{I}-\bm{H}_{\bm{\Delta}})\tilde{\bm U} \, E_{\bm a}[\bm a\bm a^\top]
\right).
\]
Since \( \bm a \) is uniformly distributed on the unit sphere in \( \mathbb R^{n-d} \), rotational invariance implies that
\[
E_{\bm a}[\bm a\bm a^\top]=\lambda \bm{I}_{n-d}
\]
for some scalar \( \lambda \). Moreover, since \( \|\bm a\|^2=1 \), we have 
\[
1
=
E_{\bm a}[\|\bm a\|^2]
=
\operatorname{tr}\!\left(E_{\bm a}[\bm a\bm a^\top]\right)
=
\lambda (n-d),
\]
consequently \( \lambda=(n-d)^{-1} \).
Therefore,
\[
E_{\bm{a}}[\bm{a}^\top \tilde{\bm{U}}^\top (\bm{I} - \bm{H}_{\bm{\Delta}})^\top (\bm{I} - \bm{H}_{\bm{\Delta}}) \tilde{\bm{U}} \bm{a}] 
= \frac{1}{n - d} \operatorname{tr} \left( \tilde{\bm{U}}^\top (\bm{I} - \bm{H}_{\bm{\Delta}})^\top (\bm{I} - \bm{H}_{\bm{\Delta}}) \tilde{\bm{U}} \right).
\]

We now simplify this trace:
\begin{align*}
\operatorname{tr}( \tilde{\bm{U}}^\top (\bm{I} - \bm{H}_{\bm{\Delta}})^\top (\bm{I} - \bm{H}_{\bm{\Delta}}) \tilde{\bm{U}} )
&= \operatorname{tr}( \tilde{\bm{U}}^\top \tilde{\bm{U}} )
- 2\operatorname{tr}( \tilde{\bm{U}}^\top \bm{H}_{\bm{\Delta}} \tilde{\bm{U}} ) 
+ \operatorname{tr}( \tilde{\bm{U}}^\top \bm{H}_{\bm{\Delta}}^\top \bm{H}_{\bm{\Delta}} \tilde{\bm{U}} ) \\
&= n - d - 2 \operatorname{tr}( \tilde{\bm{U}}^\top \bm{H}_{\bm{\Delta}} \tilde{\bm{U}} ) + \operatorname{tr}( \tilde{\bm{U}}^\top \bm{H}_{\bm{\Delta}}^\top \bm{H}_{\bm{\Delta}} \tilde{\bm{U}} ).
\end{align*}

Using the SVD \( \bm{X} = \bm{U} \bm{\Sigma} \bm{V}^\top \), we observe that
\[
\bm{H}_{\bm{\Delta}} = \bm{X}(\bm{X}^\top \bm{\Delta} \bm{X})^{-1} \bm{X}^\top \bm{\Delta} = \bm{U} (\bm{U}^\top \bm{\Delta} \bm{U})^{-1} \bm{U}^\top \bm{\Delta}.
\]
Since \( \tilde{\bm{U}}^\top \bm{U} = 0 \), we have
\[
\tilde{\bm{U}}^\top \bm{H}_{\bm{\Delta}} = \tilde{\bm{U}}^\top \bm{U} (\bm{U}^\top \bm{\Delta} \bm{U})^{-1} \bm{U}^\top \bm{\Delta} = 0,
\quad \Rightarrow \quad \operatorname{tr}( \tilde{\bm{U}}^\top \bm{H}_{\bm{\Delta}} \tilde{\bm{U}} ) = 0.
\]

It remains to compute \( \operatorname{tr}( \tilde{\bm{U}}^\top \bm{H}_{\bm{\Delta}}^\top \bm{H}_{\bm{\Delta}} \tilde{\bm{U}} ) \). Using the expression for \( \bm{H}_{\bm{\Delta}} \), we write:
\begin{align*}
\operatorname{tr}( \tilde{\bm{U}}^\top \bm{H}_{\bm{\Delta}}^\top \bm{H}_{\bm{\Delta}} \tilde{\bm{U}} )
&= \operatorname{tr}\left( \tilde{\bm{U}}^\top \bm{\Delta} \bm{U} (\bm{U}^\top \bm{\Delta} \bm{U})^{-2} \bm{U}^\top \bm{\Delta} \tilde{\bm{U}} \right) \\
&= \operatorname{tr}\left( \bm{\Delta} \bm{U} (\bm{U}^\top \bm{\Delta} \bm{U})^{-2} \bm{U}^\top \bm{\Delta} (\bm{I} - \bm{U}\bm{U}^\top) \right) \\
&= \operatorname{tr}\left( (\bm{U}^\top \bm{\Delta} \bm{U})^{-1} \right) - d.
\end{align*}

Putting all of this together:
\[
E_{\bm{a}}[\bm{a}^\top \tilde{\bm{U}}^\top (\bm{I} - \bm{H}_{\bm{\Delta}})^\top (\bm{I} - \bm{H}_{\bm{\Delta}}) \tilde{\bm{U}} \bm{a}] 
= \frac{1}{n - d} \left( n - 2d + \operatorname{tr}\left[(\bm{U}^\top \bm{\Delta} \bm{U})^{-1} \right] \right).
\]

Let
\(
\eta := \frac{\tau^2 n}{n - d} + \sigma^2,
\)
then the desirable conclusion follows.
\end{proof}
\section{Additional simulation study}
\label{app:addition_simulation}
\subsection{Simulation study with \(t\)-distributed covariates}
\label{simu:t}
In this setting, the covariate rows satisfy
\[
\bZ_i \overset{\mathrm{i.i.d.}}{\sim} t_b(5,0,{\bm \Sigma}),
\qquad
{\bm \Sigma} = 5{\bm I_b} + 5{\bm J_b},
\qquad i=1,\dots,n.
\]
All other components of the simulation setup remain the same as described in \ref{simu}.

\begin{table*}[!t]
\centering
\footnotesize
\setlength{\tabcolsep}{2.3pt}
\renewcommand{\arraystretch}{0.9}
\caption{RE\(_{\mathrm{diff}}\) and averaged selection runtime over 100 replications under multivariate \(t(5)\) covariates with covariance matrix \({\bm \Sigma} = 5{\bm I_b} + 5{\bm J_b}\). Each entry shows RE\(_{\mathrm{diff}}\) (Time). The smallest RE\(_{\mathrm{diff}}\) in each row is bolded. Hyphens (-) indicate settings that were omitted or not completed because the computations were prohibitively expensive.}
\label{tab:re_time_t5_pretest_withcoreset_sinsumsq_allcov_b15_sd1_rep100}

\begin{tabular}{ccc ccccccc}
\toprule
$n$ & $r$ & $\tau$ & SRS & Coreset & Leverage & IBOSS & ELOS & LowCon & MCG \\
\midrule

$10^{5}$ & $10^{2}$ & 0 & $.2171$ (0.00) & $.1341$ (0.05) & $.1891$ (0.02) & $.0219$ (0.04) & $\mathbf{.0180}$ (0.16) & $.1198$ (0.14) & $.1261$ (14.93) \\
$10^{5}$ & $10^{3}$ & 0 & $.0173$ (0.00) & $.0096$ (0.37) & $.0166$ (0.02) & $.0024$ (0.03) & $\mathbf{.0023}$ (0.09) & $.0107$ (0.99) & $.0165$ (15.57) \\
$10^{5}$ & $10^{4}$ & 0 & $.0015$ (0.00) & $.0011$ (4.95) & $.0017$ (0.11) & $.0004$ (0.04) & $\mathbf{.0003}$ (0.25) & $.0010$ (9.36) & $.0015$ (4.77) \\

$10^{6}$ & $10^{2}$ & 0 & $.2201$ (0.00) & $.1279$ (0.26) & $.1980$ (0.27) & $.0165$ (0.34) & $\mathbf{.0129}$ (0.92) & $.1184$ (1.78) & - \\
$10^{6}$ & $10^{3}$ & 0 & $.0170$ (0.00) & $.0090$ (4.41) & $.0160$ (0.29) & $.0016$ (0.37) & $\mathbf{.0014}$ (1.05) & $.0096$ (13.14) & - \\
$10^{6}$ & $10^{4}$ & 0 & $.0015$ (0.01) & $.0010$ (58.69) & $.0016$ (0.32) & $.0003$ (0.37) & $\mathbf{.0002}$ (1.88) & - & $.0015$ (11.28) \\

\midrule

$10^{5}$ & $10^{2}$ & 1 & $.2262$ (0.00) & $.1256$ (0.06) & $.1925$ (0.02) & $.0206$ (0.03) & $\mathbf{.0171}$ (0.13) & $.1177$ (0.15) & $.1293$ (15.41) \\
$10^{5}$ & $10^{3}$ & 1 & $.0166$ (0.00) & $.0089$ (0.37) & $.0166$ (0.02) & $.0025$ (0.03) & $\mathbf{.0023}$ (0.09) & $.0102$ (0.98) & $.0166$ (15.74) \\
$10^{5}$ & $10^{4}$ & 1 & $.0014$ (0.00) & $.0011$ (4.92) & $.0017$ (0.10) & $.0004$ (0.04) & $\mathbf{.0003}$ (0.26) & $.0010$ (9.36) & $.0015$ (4.76) \\

$10^{6}$ & $10^{2}$ & 1 & $.2315$ (0.00) & $.1279$ (0.25) & $.1952$ (0.28) & $.0154$ (0.34) & $\mathbf{.0134}$ (0.92) & $.1113$ (1.80) & - \\
$10^{6}$ & $10^{3}$ & 1 & $.0177$ (0.00) & $.0092$ (4.43) & $.0163$ (0.29) & $.0017$ (0.36) & $\mathbf{.0015}$ (1.04) & $.0090$ (13.03) & - \\
$10^{6}$ & $10^{4}$ & 1 & $.0016$ (0.01) & $.0010$ (59.00) & $.0016$ (0.31) & $.0002$ (0.37) & $\mathbf{.0002}$ (1.82) & - & $.0015$ (11.39) \\

\midrule

$10^{5}$ & $10^{2}$ & 3 & $.2047$ (0.00) & $.1265$ (0.06) & $.1753$ (0.02) & $.0194$ (0.03) & $\mathbf{.0191}$ (0.10) & $.1188$ (0.16) & $.1300$ (15.33) \\
$10^{5}$ & $10^{3}$ & 3 & $.0156$ (0.00) & $.0087$ (0.37) & $.0177$ (0.02) & $.0025$ (0.03) & $\mathbf{.0024}$ (0.10) & $.0104$ (0.99) & $.0161$ (15.99) \\
$10^{5}$ & $10^{4}$ & 3 & $.0014$ (0.00) & $.0012$ (4.85) & $.0016$ (0.11) & $.0004$ (0.04) & $\mathbf{.0003}$ (0.26) & $.0009$ (9.49) & $.0015$ (4.76) \\

$10^{6}$ & $10^{2}$ & 3 & $.2274$ (0.00) & $.1289$ (0.25) & $.1837$ (0.28) & $.0165$ (0.34) & $\mathbf{.0130}$ (0.92) & $.1079$ (1.76) & - \\
$10^{6}$ & $10^{3}$ & 3 & $.0162$ (0.00) & $.0093$ (4.54) & $.0175$ (0.29) & $.0018$ (0.37) & $\mathbf{.0016}$ (1.05) & $.0091$ (12.99) & - \\
$10^{6}$ & $10^{4}$ & 3 & $.0016$ (0.01) & $.0010$ (114.76) & $.0016$ (0.31) & $.0003$ (0.37) & $\mathbf{.0002}$ (2.25) & - & $.0015$ (10.83) \\

\midrule

$10^{5}$ & $10^{2}$ & 5 & $.2038$ (0.00) & $.1273$ (0.04) & $.1872$ (0.02) & $.0211$ (0.04) & $\mathbf{.0153}$ (0.10) & $.1197$ (0.15) & $.1297$ (17.27) \\
$10^{5}$ & $10^{3}$ & 5 & $.0172$ (0.00) & $.0091$ (0.37) & $.0153$ (0.02) & $.0025$ (0.03) & $\mathbf{.0024}$ (0.10) & $.0105$ (0.99) & $.0160$ (15.77) \\
$10^{5}$ & $10^{4}$ & 5 & $.0014$ (0.00) & $.0011$ (4.79) & $.0017$ (0.10) & $.0004$ (0.04) & $\mathbf{.0003}$ (0.26) & $.0009$ (9.42) & $.0015$ (4.81) \\

$10^{6}$ & $10^{2}$ & 5 & $.2262$ (0.00) & $.1323$ (0.26) & $.1991$ (0.28) & $.0159$ (0.36) & $\mathbf{.0130}$ (0.94) & $.1081$ (1.75) & - \\
$10^{6}$ & $10^{3}$ & 5 & $.0167$ (0.00) & $.0088$ (4.55) & $.0170$ (0.30) & $.0015$ (0.37) & $\mathbf{.0014}$ (1.07) & $.0092$ (13.23) & - \\
$10^{6}$ & $10^{4}$ & 5 & $.0016$ (0.01) & $.0011$ (100.58) & $.0016$ (0.38) & $.0003$ (0.53) & $\mathbf{.0002}$ (2.31) & - & $.0015$ (10.64) \\

\bottomrule
\end{tabular}
\end{table*}


Table~\ref{tab:re_time_t5_pretest_withcoreset_sinsumsq_allcov_b15_sd1_rep100} shows that IBOSS and ELOS consistently achieve the lowest RE\(_{\mathrm{diff}}\) across all settings, with ELOS performing comparably to or better than IBOSS in all cases. SRS, LEV, LowCon and MCG show considerably higher relative errors, especially at $r=100$. Coreset performs better than SRS and LEV but remains well above IBOSS and ELOS throughout. LowCon also improves over SRS and LEV in many settings, but its RE\(_{\mathrm{diff}}\) remains substantially larger than those of IBOSS and ELOS. In terms of selection runtime, SRS is the fastest, followed by leverage scoring and IBOSS, as these methods are explicitly designed to balance computational cost and statistical efficiency. MCG and LowCon incur substantially higher runtime than all other methods.  The runtimes of Coreset and ELOS are mixed: Coreset is faster at small $r$, but its runtime grows substantially with $r$, whereas ELOS scales more favorably as $r$ increases.

For the MCG method, we ran the same \((n,r,\tau)\) combinations described in \ref{simu} for a direct performance comparison. The results indicate that the measurement-constrained gradient method is sensitive to the covariate distribution: under heavy-tailed covariates, both its predictive performance and its computational cost can deteriorate noticeably. In particular, the runtime appears to depend not only on the problem size, but also on the behavior of the internal optimization routine, such as the number of iterations required by the algorithm. We also experimented with larger iteration caps (500--1000) in some cases, but found that this substantially increased the computational cost while yielding little improvement in RE. This observation is in line with the discussion in \citet[Proposition 21]{WangYuSingh2017}, who remark that the theoretical convergence rate is conservative and that practical convergence is often achieved in only a few iterations, especially when the budget \(r\) is large. Additional diagnostic summaries are reported in Appendix~\ref{appendix:diag} Table~\ref{tab:measurement_constrained_diag_t_by_tau}.

\subsection{Simulation study with normal distributed covariates}
\label{app:simulation-boxplots}

The simulation Table~\ref{tab:re_time_normal_withcoreset_sinsumsq_allcov_b15_sd1_rep100} reports averaged RE\(_{\mathrm{diff}}\) and selection runtime over 100 replications for each experimental setting. To provide a more complete view of the variability across replications, we include boxplots in Figure~\ref{fig:normal_boxplot_n1e5} of RE\(_{\mathrm{diff}}\).
Figure~\ref{fig:normal_boxplot_n1e5} shows the empirical distributions of RE\(_{\mathrm{diff}}\) across 100 replications under multivariate normal covariates with \(n=10^{5}\). Each panel corresponds to one \((\tau,r)\) setting, and each boxplot summarizes the distribution of RE\(_{\mathrm{diff}}\) for one method across replications. 

Overall, the boxplots show ELOS generally exhibits small RE\(_{\mathrm{diff}}\) and relatively concentrated distributions, especially for smaller subset sizes. Coreset also performs strongly in several settings, particularly when \(r\) is large, but requires substantially larger selection time as reported in the main table.
\begin{table*}[!t]
\centering
\setlength{\tabcolsep}{2.2pt}
\renewcommand{\arraystretch}{0.75}
\caption{Averaged RE\(_{\mathrm{diff}}\) and selection runtime over 100 replications under multivariate normal covariates. Each entry shows RE\(_{\mathrm{diff}}\) (Time). The smallest RE\(_{\mathrm{diff}}\) in each row is bolded. The full variability of RE\(_{\mathrm{diff}}\) is visualized by boxplots in Figure~\ref{fig:normal_boxplot_n1e5}. Hyphens (-) indicate settings that were omitted or not completed because the computations were prohibitively expensive.}
\label{tab:re_time_normal_withcoreset_sinsumsq_allcov_b15_sd1_rep100}

{\footnotesize
\begin{tabular}{ccc ccccccc}
\toprule
$n$ & $r$ & $\tau$ & SRS & Coreset & Leverage & IBOSS & ELOS & LowCon & MCG \\
\midrule

$10^{5}$ & $10^{2}$ & 0 & $.1857$ (0.00) & $.1051$ (0.10) & $.1863$ (0.02) & $.1235$ (0.03) & $\mathbf{.0703}$ (0.11) & $.1615$ (0.28) & $.1594$ (15.08) \\
$10^{5}$ & $10^{3}$ & 0 & $.0160$ (0.00) & $\mathbf{.0044}$ (0.28) & $.0166$ (0.06) & $.0107$ (0.03) & $.0075$ (0.13) & $.0146$ (1.14) & $.0162$ (1.68) \\
$10^{5}$ & $10^{4}$ & 0 & $.0014$ (0.00) & $\mathbf{.0004}$ (4.66) & $.0016$ (0.15) & $.0011$ (0.04) & $.0009$ (0.11) & $.0012$ (10.02) & $.0016$ (9.77) \\

$10^{6}$ & $10^{2}$ & 0 & $.1909$ (0.00) & $.1064$ (0.27) & $.1769$ (0.26) & $.1134$ (0.34) & $\mathbf{.0686}$ (0.92) & $.1671$ (3.40) & - \\
$10^{6}$ & $10^{3}$ & 0 & $.0158$ (0.00) & $\mathbf{.0044}$ (4.53) & $.0149$ (0.28) & $.0103$ (0.36) & $.0065$ (1.02) & $.0131$ (13.12) & - \\
$10^{6}$ & $10^{4}$ & 0 & $.0015$ (0.01) & $\mathbf{.0003}$ (58.58) & $.0017$ (0.34) & $.0011$ (0.38) & $.0007$ (1.87) & $.0014$ (122.39) & $.0016$ (10.49) \\

\midrule

$10^{5}$ & $10^{2}$ & 1 & $.1947$ (0.00) & $.1016$ (0.05) & $.1876$ (0.02) & $.1107$ (0.03) & $\mathbf{.0683}$ (0.12) & $.1765$ (0.24) & $.1584$ (15.39) \\
$10^{5}$ & $10^{3}$ & 1 & $.0155$ (0.00) & $\mathbf{.0043}$ (0.27) & $.0164$ (0.06) & $.0117$ (0.03) & $.0075$ (0.13) & $.0140$ (1.14) & $.0158$ (1.69) \\
$10^{5}$ & $10^{4}$ & 1 & $.0015$ (0.00) & $\mathbf{.0003}$ (4.58) & $.0015$ (0.15) & $.0010$ (0.04) & $.0008$ (0.11) & $.0013$ (10.40) & $.0016$ (9.99) \\

$10^{6}$ & $10^{2}$ & 1 & $.1982$ (0.00) & $.1043$ (0.26) & $.1902$ (0.27) & $.1033$ (0.34) & $\mathbf{.0655}$ (0.91) & $.1748$ (3.17) & - \\
$10^{6}$ & $10^{3}$ & 1 & $.0161$ (0.00) & $\mathbf{.0042}$ (12.14) & $.0163$ (0.30) & $.0099$ (0.37) & $.0061$ (1.03) & $.0132$ (14.29) & - \\
$10^{6}$ & $10^{4}$ & 1 & $.0016$ (0.01) & $\mathbf{.0003}$ (87.14) & $.0016$ (0.32) & $.0011$ (0.37) & $.0008$ (1.85) & $.0014$ (128.47) & $.0015$ (10.56) \\

\midrule

$10^{5}$ & $10^{2}$ & 3 & $.2022$ (0.00) & $.1044$ (0.11) & $.1965$ (0.01) & $.1326$ (0.03) & $\mathbf{.0705}$ (0.11) & $.1843$ (0.26) & $.1556$ (15.47) \\
$10^{5}$ & $10^{3}$ & 3 & $.0154$ (0.00) & $\mathbf{.0037}$ (0.27) & $.0164$ (0.06) & $.0107$ (0.03) & $.0077$ (0.13) & $.0144$ (1.15) & $.0159$ (1.70) \\
$10^{5}$ & $10^{4}$ & 3 & $.0014$ (0.00) & $\mathbf{.0003}$ (4.89) & $.0016$ (0.15) & $.0011$ (0.04) & $.0008$ (0.11) & $.0014$ (10.38) & $.0015$ (9.99) \\

$10^{6}$ & $10^{2}$ & 3 & $.1943$ (0.00) & $.1019$ (0.27) & $.1904$ (0.27) & $.1143$ (0.34) & $\mathbf{.0661}$ (0.91) & $.1798$ (3.34) & - \\
$10^{6}$ & $10^{3}$ & 3 & $.0170$ (0.00) & $\mathbf{.0038}$ (4.66) & $.0175$ (0.30) & $.0109$ (0.37) & $.0066$ (1.07) & $.0134$ (14.93) & - \\
$10^{6}$ & $10^{4}$ & 3 & $.0017$ (0.01) & $\mathbf{.0003}$ (97.27) & $.0016$ (0.32) & $.0011$ (0.38) & $.0008$ (1.97) & $.0014$ (128.46) & $.0015$ (10.50) \\

\midrule

$10^{5}$ & $10^{2}$ & 5 & $.1906$ (0.00) & $.0979$ (0.13) & $.1912$ (0.01) & $.1160$ (0.03) & $\mathbf{.0723}$ (0.10) & $.1841$ (0.26) & $.1548$ (15.08) \\
$10^{5}$ & $10^{3}$ & 5 & $.0171$ (0.00) & $\mathbf{.0036}$ (0.28) & $.0162$ (0.06) & $.0115$ (0.03) & $.0076$ (0.13) & $.0146$ (1.12) & $.0160$ (1.67) \\
$10^{5}$ & $10^{4}$ & 5 & $.0014$ (0.00) & $\mathbf{.0003}$ (14.22) & $.0016$ (0.16) & $.0011$ (0.04) & $.0009$ (0.12) & $.0014$ (9.37) & $.0015$ (9.90) \\

$10^{6}$ & $10^{2}$ & 5 & $.1888$ (0.00) & $.0963$ (0.28) & $.1921$ (0.28) & $.1139$ (0.34) & $\mathbf{.0640}$ (0.92) & $.1801$ (3.13) & - \\
$10^{6}$ & $10^{3}$ & 5 & $.0156$ (0.00) & $\mathbf{.0040}$ (4.76) & $.0165$ (0.31) & $.0101$ (0.39) & $.0069$ (1.12) & $.0135$ (14.40) & - \\
$10^{6}$ & $10^{4}$ & 5 & $.0017$ (0.01) & $\mathbf{.0003}$ (82.26) & $.0016$ (0.31) & $.0011$ (0.36) & $.0008$ (2.26) & - & $.0015$ (10.86) \\

\bottomrule
\end{tabular}
}
\end{table*}

\begin{figure*}[!t]
    \centering
    \includegraphics[width=0.88\textwidth]{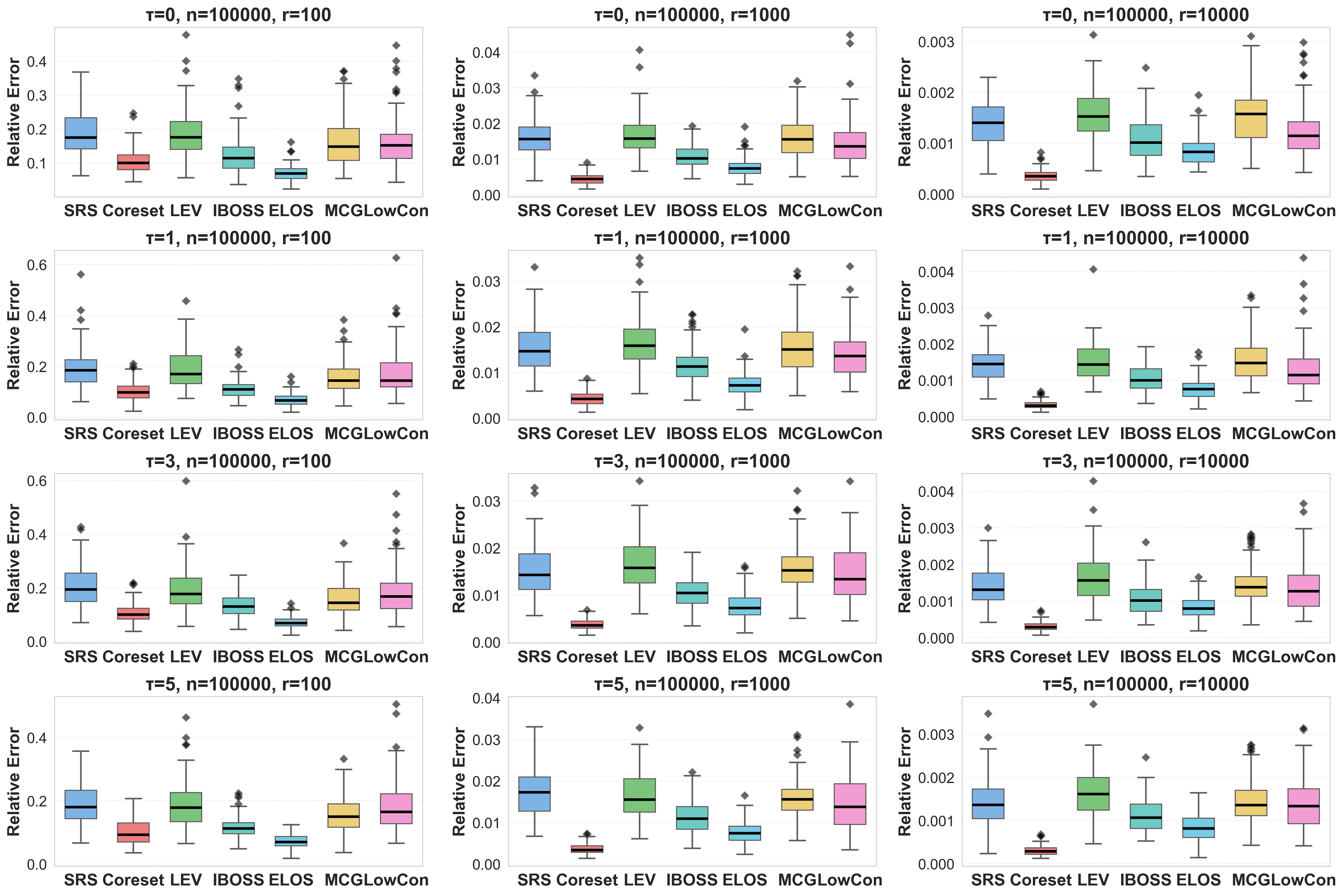}
    \caption{
    Empirical distributions of RE\(_{\mathrm{diff}}\) across 100 replications under multivariate normal covariates with \(n=10^{5}\). The corresponding averaged RE\(_{\mathrm{diff}}\) and selection runtimes are reported in Table~\ref{tab:re_time_normal_withcoreset_sinsumsq_allcov_b15_sd1_rep100}. Each panel corresponds to a fixed \((\tau,r)\) setting, and each boxplot summarizes the variation of one method across replications. This figure complements the main table by showing the spread, stability, and potential outliers of the competing methods.
    }
    \label{fig:normal_boxplot_n1e5}
\end{figure*}

\section{Diagnostic tables for MCG}
\label{appendix:diag}

Table~\ref{tab:measurement_constrained_diag_normal_by_tau} and Table~\ref{tab:measurement_constrained_diag_t_by_tau} report additional diagnostic summaries for the measurement-constrained gradient (MCG) method under the simulation settings considered in the main text. In addition to the relative error difference and runtime, the tables include several optimization-related quantities, such as the mean number of outer iterations, the iteration-cap hit rate, the average backtracking count, the average number of projection iterations, and the final relative improvement. These results provide additional insight into the computational behavior of MCG across different problem settings.
\begin{table*}[!t]
\centering
\footnotesize
\caption{Diagnostics of the measurement-constrained method under multivariate normal covariates.}
\label{tab:measurement_constrained_diag_normal_by_tau}
\setlength{\tabcolsep}{3pt}
\renewcommand{\arraystretch}{0.75}
\begin{tabular}{cccccccccc}
\toprule
$n$ & $r$ & $\tau$ & Mean RE &
\begin{tabular}[c]{@{}c@{}}Time\\ (s)\end{tabular} &
\begin{tabular}[c]{@{}c@{}}Mean Outer\\ Iter\end{tabular} &
\begin{tabular}[c]{@{}c@{}}Cap Hit\\ Rate\end{tabular} &
\begin{tabular}[c]{@{}c@{}}Mean Avg\\ BT\end{tabular} &
\begin{tabular}[c]{@{}c@{}}Avg Proj\\ Iter\end{tabular} &
\begin{tabular}[c]{@{}c@{}}Final Rel\\ Improve\end{tabular} \\
\midrule
$10^5$ & $10^2$ & 0 & 0.1594 & 15.08 & 20 & 1.0 & 1.0 & 12.313 & 5.10e-03 \\
$10^5$ & $10^2$ & 1 & 0.1584 & 15.39 & 20 & 1.0 & 1.0 & 12.313 & 5.10e-03 \\
$10^5$ & $10^2$ & 3 & 0.1556 & 15.47 & 20 & 1.0 & 1.0 & 12.313 & 5.10e-03 \\
$10^5$ & $10^2$ & 5 & 0.1548 & 15.08 & 20 & 1.0 & 1.0 & 12.313 & 5.10e-03 \\
\addlinespace
$10^5$ & $10^3$ & 0 & 0.0162 & 1.68 & 2 & 0.0 & 1.0 & 14.000 & 6.80e-06 \\
$10^5$ & $10^3$ & 1 & 0.0164 & 1.59 & 2 & 0.0 & 1.0 & 14.000 & 6.80e-06 \\
$10^5$ & $10^3$ & 3 & 0.0149 & 1.56 & 2 & 0.0 & 1.0 & 14.000 & 6.80e-06 \\
$10^5$ & $10^3$ & 5 & 0.0146 & 1.69 & 2 & 0.0 & 1.0 & 14.000 & 6.80e-06 \\
\addlinespace
$10^5$ & $10^4$ & 0 & 0.0014 & 9.64 & 1 & 0.0 & 50.0 & 950.000 & -9.54e-07 \\
$10^5$ & $10^4$ & 1 & 0.0014 & 9.53 & 1 & 0.0 & 50.0 & 950.000 & -9.54e-07 \\
$10^5$ & $10^4$ & 3 & 0.0014 & 9.57 & 1 & 0.0 & 50.0 & 950.000 & -9.54e-07 \\
$10^5$ & $10^4$ & 5 & 0.0014 & 9.59 & 1 & 0.0 & 50.0 & 950.000 & -9.54e-07 \\
\addlinespace
$10^6$ & $10^4$ & 0 & 0.0015 & 5.47 & 1 & 0.0 & 1.0 & 18.000 & 1.46e-07 \\
$10^6$ & $10^4$ & 1 & 0.0015 & 5.48 & 1 & 0.0 & 1.0 & 18.000 & 1.46e-07 \\
$10^6$ & $10^4$ & 3 & 0.0014 & 5.52 & 1 & 0.0 & 1.0 & 18.000 & 1.46e-07 \\
$10^6$ & $10^4$ & 5 & 0.0015 & 5.51 & 1 & 0.0 & 1.0 & 18.000 & 1.46e-07 \\
\bottomrule
\end{tabular}
\end{table*}

\begin{table*}[!t]
\centering
\footnotesize
\caption{Diagnostics of the measurement-constrained method under multivariate \(t(5)\) covariates.}
\label{tab:measurement_constrained_diag_t_by_tau}
\setlength{\tabcolsep}{3pt}
\renewcommand{\arraystretch}{0.75}
\begin{tabular}{cccccccccc}
\toprule
$n$ & $r$ & $\tau$ & Mean RE &
\begin{tabular}[c]{@{}c@{}}Time\\ (s)\end{tabular} &
\begin{tabular}[c]{@{}c@{}}Mean Outer\\ Iter\end{tabular} &
\begin{tabular}[c]{@{}c@{}}Cap Hit\\ Rate\end{tabular} &
\begin{tabular}[c]{@{}c@{}}Mean Avg\\ BT\end{tabular} &
\begin{tabular}[c]{@{}c@{}}Avg Proj\\ Iter\end{tabular} &
\begin{tabular}[c]{@{}c@{}}Final Rel\\ Improve\end{tabular} \\
\midrule
$10^5$ & $10^2$ & 0 & 0.1261 & 14.93 & 20 & 1.0 & 1.001 & 15.4685 & 6.44e-3 \\
$10^5$ & $10^2$ & 1 & 0.1293 & 15.41 & 20 & 1.0 & 1.001 & 15.4685 & 6.44e-3 \\
$10^5$ & $10^2$ & 3 & 0.1300 & 15.32 & 20 & 1.0 & 1.001 & 15.4685 & 6.44e-3 \\
$10^5$ & $10^2$ & 5 & 0.1297 & 17.27 & 20 & 1.0 & 1.001 & 15.4685 & 6.44e-3 \\
\addlinespace
$10^5$ & $10^3$ & 0 & 0.0165 & 15.57 & 20 & 1.0 & 1.000 & 13.2660 & 1.16e-4 \\
$10^5$ & $10^3$ & 1 & 0.0166 & 15.74 & 20 & 1.0 & 1.000 & 13.2660 & 1.16e-4 \\
$10^5$ & $10^3$ & 3 & 0.0161 & 15.99 & 20 & 1.0 & 1.000 & 13.2660 & 1.16e-4 \\
$10^5$ & $10^3$ & 5 & 0.0160 & 15.77 & 20 & 1.0 & 1.000 & 13.2660 & 1.16e-4 \\
\addlinespace
$10^5$ & $10^4$ & 0 & 0.0015 & 4.77 & 1 & 0.0 & 20.000 & 380.0000 & -9.54e-7 \\
$10^5$ & $10^4$ & 1 & 0.0015 & 4.76 & 1 & 0.0 & 20.000 & 380.0000 & -9.54e-7 \\
$10^5$ & $10^4$ & 3 & 0.0015 & 4.76 & 1 & 0.0 & 20.000 & 380.0000 & -9.54e-7 \\
$10^5$ & $10^4$ & 5 & 0.0015 & 4.81 & 1 & 0.0 & 20.000 & 380.0000 & -9.54e-7 \\
\addlinespace
$10^6$ & $10^4$ & 0 & 0.0015 & 11.27 & 1 & 0.0 & 1.000 & 18.0000 & 6.57e-7 \\
$10^6$ & $10^4$ & 1 & 0.0015 & 11.38 & 1 & 0.0 & 1.000 & 18.0000 & 6.57e-7 \\
$10^6$ & $10^4$ & 3 & 0.0015 & 10.83 & 1 & 0.0 & 1.000 & 18.0000 & 6.57e-7 \\
$10^6$ & $10^4$ & 5 & 0.0015 & 10.64 & 1 & 0.0 & 1.000 & 18.0000 & 6.57e-7 \\
\bottomrule
\end{tabular}
\end{table*}

\section{Additional details for real-data experiments}
\label{app:real_data_details}

\subsection{Residual bootstrap procedure}
\label{app:bootstrap}
The residual bootstrap is a classical tool for regression inference. Its construction and validity for least-squares regression have been studied in the regression-bootstrap literature, including \citet{Freedman1981}, empirical investigations such as \citet{FreedmanPeters1984}, and standard treatments in \citet{EfronTibshirani1993,DavisonHinkley1997}. Related residual-bootstrap ideas have also been considered in modern large-scale regression; for example, \citet{GangulySenguptaGhosh2024} proposed a subsampled residual bootstrap for generalized linear models to reduce the cost of full-size bootstrap resampling.

Specifically, for each real dataset, we first fit the full-data OLS estimator and compute
\[
\hat{\bm y}=\bX\hat{\bm\beta}_{\mathrm{full}},
\qquad
\hat{\bm e}=\bm y-\hat{\bm y}.
\]
The residuals are centered as \(\tilde e_i=\hat e_i-\bar{\hat e}\), then resampled with replacement to obtain \(\tilde e_i^*\). The bootstrap pseudo-response is defined as
\[
\bm y^*=\hat{\bm y}+\tilde{\bm e}^*.
\]
\added{Because the centered residuals are resampled independently of the covariates, this residual bootstrap does not preserve residual--covariate dependence. If an important component of the real-data misspecification is carried by such dependence, the bootstrap experiment may understate that aspect of structured misspecification.}

We adopt the residual bootstrap to generate pseudo-responses for the real-data evaluation. One may wonder whether this approach inadvertently masks misspecification by construction. However, committing to a specific parametric form of misspecification would evaluate methods under one arbitrarily chosen departure from linearity, which is no less restrictive. In practice, before responses are collected, the nature of any structural departure from the linear model — whether nonlinear dynamics, omitted interactions, or other mechanisms — is unknown. The residual bootstrap instead draws resamples from the empirical residual distribution, which implicitly captures whatever unexplained structure remains after fitting the linear model to the observed data. To the extent that the original data contain systematic departures from linearity, these are partially reflected in the residuals and propagated into the bootstrap samples. The residual bootstrap therefore provides a data-driven, assumption-lean way to mimic the range of data-generating mechanisms consistent with the observed data, without committing to a single misspecification scenario — a more honest reflection of the uncertainty a practitioner faces in real applications.

\section{Training and testing split results}
\label{app:train_test}

\subsection{Chemical Sensor Data}
\label{test_train_sensor}

The chemical sensors dataset was generated at the ChemoSignals Laboratory within the BioCircuits Institute at the University of California, San Diego. It contains measurements from an array of 16 chemical sensors exposed to air mixtures containing Ethylene and CO at varying concentration levels. Each experiment involved continuous signal acquisition over approximately 12 hours, during which concentration changes occurred at randomly selected times and magnitudes. Additional details are provided in \cite{fonollosa2015dataset}.

Following the preprocessing protocol of \cite{wang2019information}, we discard the first 20{,}000 observations corresponding to system run-in time, use the last sensor reading as the response, exclude the second sensor due to negative readings, and apply a log transformation to the remaining variables. A linear regression model with intercept is then fitted on the transformed scale. After preprocessing, the analyzed design has \(n=4{,}188{,}261\) observations and \(p=14\) covariates plus an intercept, matching the dimensions reported in \cite{wang2019information}. Since the underlying sensor-response relationship may involve nonlinear dynamics, interactions, or other complex effects, the linear model should be viewed as a working approximation.
\begin{table*}[!t]
\centering
\footnotesize
\setlength{\tabcolsep}{2.5pt}
\renewcommand{\arraystretch}{0.9}
\renewcommand{\arraystretch}{1}
\caption{Residual bootstrap results on the Chemical Sensor dataset. We report RE\(_{\mathrm{diff}}\) with respect to full-data OLS and total selection runtime (seconds), averaged over 100 bootstrap replicates, with standard deviations in parentheses. Runtime includes all method-specific preprocessing required for subset construction, including the SVD used by ELOS and the augmented-matrix SVD used by coreset. The smallest RE\(_{\mathrm{diff}}\) in each \(r\) block is bolded. LowCon for \(r=5{,}000\) and MCG for \(r=1{,}000\) were omitted because the corresponding computations were prohibitively expensive and preliminary runs did not show competitive predictive performance.}
\label{tab:sensor_realdata_residboot}
\begin{tabular}{lcc@{\hspace{8mm}}cc}
\toprule
& \multicolumn{2}{c}{$r = 1{,}000$} & \multicolumn{2}{c}{$r = 5{,}000$} \\
\cmidrule(lr){2-3}\cmidrule(lr){4-5}
Method & RE(Std.) & Sel. Time (s) & RE(Std.) & Sel. Time (s) \\
\midrule
SRS      & 0.0158 (0.0062) & 0.0224 (0.0022) & 0.0029 (0.0010) & 0.0221 (0.0056) \\
Coreset  & 0.0137 (0.0068) & 34.2633 (3.4782) & 0.0039 (0.0017) & 179.3040 (13.4508) \\
Leverage & 0.0156 (0.0064) & 1.0293 (0.0967) & 0.0030 (0.0013) & 0.8829 (0.0574) \\
IBOSS    & 0.0521 (0.0325) & 0.9884 (0.0738) & 0.0098 (0.0069) & 1.2138 (0.1142) \\
ELOS     & \textbf{0.0073 (0.0031)} & 4.0822 (0.4192) & \textbf{0.0017 (0.0009)} & 4.7520 (1.0088) \\
LowCon   & 0.0224 (0.0193) & 106.7156 (11.2301) & -- & -- \\
MCG      & -- & -- & 0.0024 (0.0011) & 101.5206 (4.5337) \\
\bottomrule
\end{tabular}
\end{table*}
To further assess generalization beyond the fixed-design residual-bootstrap comparison, we also evaluate each method using a fixed held-out test set of 10,000 observations in Table~\ref{tab:sensor_realdata_residboot1}. 
The train and test RE values show consistent patterns across methods and subsampling size r, with relative differences typically small, indicating that the observed performance differences are not driven by overfitting to 
the training portion of the design. The ranking of methods is also consistent with the full-data evaluation: ELOS delivers the smallest test RE at both \(r=1{,}000\) and \(r=5{,}000\), whereas IBOSS again shows the largest performance gap, reinforcing the conclusion that ELOS provides improved robustness on this misspecified real-data design.

\begin{table*}[!t]
\centering
\footnotesize
\setlength{\tabcolsep}{2pt}
\renewcommand{\arraystretch}{1}
\caption{Residual bootstrap results on the Chemical Sensor dataset. We report training and test RE\(_{\mathrm{diff}}\) and runtime (seconds), averaged over \(B=50\) completed replicates, with standard deviations in parentheses. Runtime includes all method-specific preprocessing required for subset construction, including the SVD used by ELOS and the augmented-matrix SVD used by coreset. The smallest test RE\(_{\mathrm{diff}}\) in each \(r\) block is bolded. LowCon and MCG were omitted from this comparison because preliminary runs showed limited predictive improvement while requiring substantially higher computational time.}
\label{tab:sensor_realdata_residboot1}
\begin{tabular}{lccc@{\hspace{5mm}}ccc}
\toprule
& \multicolumn{3}{c}{$r = 1{,}000$} & \multicolumn{3}{c}{$r = 5{,}000$} \\
\cmidrule(lr){2-4}\cmidrule(lr){5-7}
Method & Train & Test & Time & Train & Test & Time \\
\midrule
SRS      & .0151 (.0068) & .0161 (.0076) & 0.02 (0.00) & .0030 (.0008) & .0030 (.0013) & 0.02 (0.00) \\
Coreset  & .0145 (.0063) & .0141 (.0073) & 18.35 (1.20) & .0049 (.0023) & .0053 (.0026) & 93.28 (4.68) \\
Leverage & .0185 (.0062) & .0188 (.0073) & 1.11 (0.15) & .0030 (.0008) & .0031 (.0013) & 0.95 (0.06) \\
IBOSS    & .0405 (.0242) & .0422 (.0257) & 1.10 (0.18) & .0084 (.0057) & .0089 (.0068) & 1.73 (0.07) \\
ELOS     & .0076 (.0025) & \textbf{.0078} (.0028) & 4.92 (0.32) & .0018 (.0009) & \textbf{.0020} (.0014) & 5.57 (0.50) \\
\bottomrule
\end{tabular}
\end{table*}
\subsection{Million Song Dataset}
\label{test_train_song}

The YearPredictionMSD dataset\citep{year_prediction_msd_203} was constructed from the Million Song Dataset and is available from the UCI Machine Learning Repository. It contains \(n=515{,}345\) tracks, each represented by \(p=90\) acoustic descriptors derived from timbre and spectral characteristics. The release year of each track is used as the response variable.

This dataset is commonly used in large-scale regression and kernel-learning studies, where flexible nonlinear models are often employed to capture complex relationships between acoustic features and temporal information. In contrast, we use a linear regression working model. Given the scale and complexity of the data, this linear specification should not be interpreted as correctly specified, but rather as a deliberately restricted working approximation. The dataset therefore provides a realistic large-scale setting for evaluating subdata selection under possible model misspecification.

Before subdata selection, we apply feature screening to all methods to evaluate their performance under different dimensional settings. Specifically, for the residual-bootstrap comparison in Table~\ref{tab:year_msd_residboot_noscr}, we retain the top \(k\in\{30,60,90\}\) predictors ranked by absolute Pearson correlation \(|\mathrm{corr}(\by,\bX_j)|\) and form the screened design matrix
\[
\bA=[\mathbf{1},\,\bX_{\mathrm{screen}}].
\]

The $k=30$ and $k=60$ panels in Figure~\ref{fig:boxplots_yearmsd}  use response-dependent Pearson screening only to define common reduced predictor sets, they are shown as conditional dimensionality and computational-scaling comparisons. The $k=90$ panels use all original predictors and require no response-dependent screening.

We also conduct a train/test evaluation using a held-out test set of size \(10{,}000\). 
Similarly, we apply a common feature screening step to define the feature set for each value of \(k\). Specifically, features are ranked by the absolute Pearson correlations with the response, and the same selected feature set is used for both training and held-out test sets. This screening step is used only to define the feature dimension under comparison.


Table~\ref{tab:year_msd_residboot_tt} shows that the train and test RE values are close across methods, suggesting that the relative performance differences are stable beyond the training design. ELOS attains the smallest test RE\(_{\mathrm{diff}}\) in the higher-dimensional settings.

\begin{figure*}[!t]
    \centering
    \includegraphics[width=0.88\textwidth]{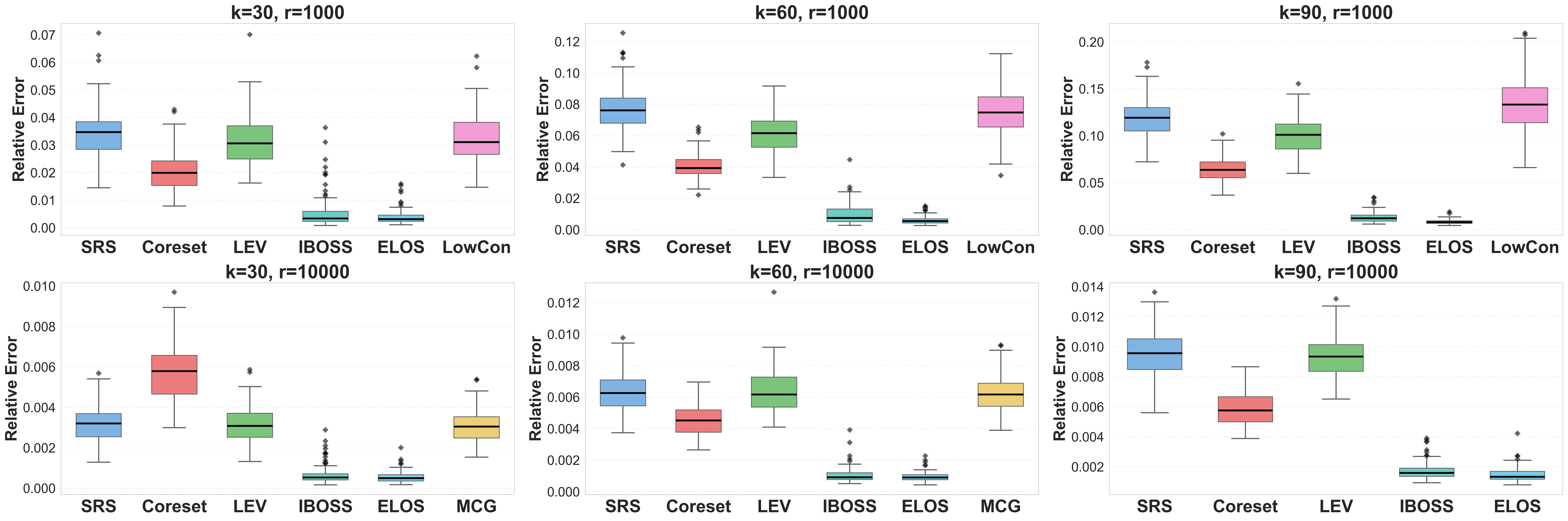}
    \caption{Boxplots of RE\(_{\mathrm{diff}}\) across 100 bootstrap replicates for the million song data, organized by number of predictors k and subsample size r. Columns represent different subsample sizes, rows represent different numbers of predictors. Table~\ref{tab:year_msd_residboot_noscr} in Appendix~\ref{test_train_song} gives averaged RE\(_{\mathrm{diff}}\), computational time, and standard deviations.}
    \label{fig:boxplots_yearmsd}
\end{figure*}

\begin{table*}[!t]
\centering
\footnotesize
\renewcommand{\arraystretch}{0.9}
\setlength{\tabcolsep}{0.7pt}
\caption{Residual bootstrap results on the YearPredictionMSD dataset: relative error (RE) with respect to full-data OLS and selection runtime (seconds). RE is reported as mean (standard deviation) over 100 bootstrap replicates. For each $(r,k)$ block, the smallest mean RE is bolded. Hyphens (-) indicate settings omitted because the computations were prohibitively expensive; in particular, LowCon for \(r=10{,}000\) was not completed due to its substantially higher runtime.}
\label{tab:year_msd_residboot_noscr}
\begin{tabular}{lcccccc}
\toprule
& \multicolumn{2}{c}{$k = 30$} & \multicolumn{2}{c}{$k = 60$} & \multicolumn{2}{c}{$k = 90$} \\
\cmidrule(lr){2-3}\cmidrule(lr){4-5}\cmidrule(lr){6-7}
Method & RE (Std.) & Sel. Time (s) & RE (Std.) & Sel. Time (s) & RE (Std.) & Sel. Time (s) \\
\midrule
\multicolumn{7}{l}{\textbf{Subdata size:} $r = 1{,}000$} \\
SRS      & 0.0349 (0.0094) & 0.0001 & 0.0767 (0.0144) & 0.0001 & 0.1197 (0.0194) & 0.0001 \\
Coreset  & 0.0202 (0.0066) & 2.4631 & 0.0374 (0.0078) & 6.2496 & 0.0615 (0.0116) & 15.3861 \\
Leverage & 0.0319 (0.0091) & 0.1847 & 0.0626 (0.0129) & 0.4967 & 0.0972 (0.0206) & 1.0365 \\
IBOSS    & 0.0056 (0.0061) & 0.4557 & 0.0083 (0.0065) & 0.7674 & 0.0121 (0.0055) & 1.2238 \\
ELOS     & \textbf{0.0041} (0.0027) & 2.4850 & \textbf{0.0061} (0.0026) & 9.4978 & \textbf{0.0083} (0.0025) & 28.1364 \\
LowCon   & 0.0325 (0.0094) & 9.0637 & 0.0757 (0.0146) & 17.2030 & 0.1356 (0.0281) & 22.6801 \\
MCG      & -- & -- & -- & -- & -- & -- \\
\midrule
\multicolumn{7}{l}{\textbf{Subdata size:} $r = 10{,}000$} \\
SRS      & 0.0029 (0.0009) & 0.0002 & 0.0061 (0.0012) & 0.0002 & 0.0094 (0.0015) & 0.0002 \\
Coreset  & 0.0064 (0.0014) & 34.5773 & 0.0047 (0.0010) & 63.5551 & 0.0064 (0.0010) & 132.1411 \\
Leverage & 0.0031 (0.0009) & 0.3171 & 0.0064 (0.0014) & 0.5865 & 0.0093 (0.0014) & 1.0988 \\
IBOSS    & 0.0007 (0.0005) & 0.3579 & 0.0011 (0.0005) & 0.7190 & 0.0017 (0.0006) & 1.1905 \\
ELOS     & \textbf{0.0006} (0.0003) & 2.8058 & \textbf{0.0010} (0.0003) & 9.8884 & \textbf{0.0015} (0.0005) & 28.9546 \\
LowCon   & -- & -- & -- & -- & -- & -- \\
MCG      & 0.0031 (0.0008) & 87.3603 & 0.0062 (0.0012) & 200.2522 & -- & -- \\
\bottomrule
\end{tabular}
\end{table*}

\begin{table*}[!t] \centering
\footnotesize \setlength{\tabcolsep}{5pt} \renewcommand{\arraystretch}{0.95} \caption{Residual bootstrap results on the YearPredictionMSD dataset. We report training and test RE\(_{\mathrm{diff}}\) and selection time (seconds), averaged over \(B=50\) replicates. For each \((r,k)\) block, the smallest train and test RE\(_{\mathrm{diff}}\) are bolded. LowCon and MCG were omitted from this comparison because preliminary runs showed limited predictive improvement while requiring substantially higher computational time.} \label{tab:year_msd_residboot_tt} \begin{tabular}{lccccccccc} \toprule & \multicolumn{3}{c}{$k = 30$} & \multicolumn{3}{c}{$k = 60$} & \multicolumn{3}{c}{$k = 90$} \\ \cmidrule(lr){2-4}\cmidrule(lr){5-7}\cmidrule(lr){8-10} Method & Train & Test & Time & Train & Test & Time & Train & Test & Time \\ \midrule \multicolumn{10}{l}{\textbf{$r = 1{,}000$}} \\ SRS & .0343 & .0338 & 0.0001 & .0751 & .0746 & 0.0001 & .1213 & .1220 & 0.0000 \\ Coreset & .0211 & .0206 & 3.0470 & .0428 & .0428 & 7.3525 & .0631 & .0638 & 14.9176 \\ Leverage & .0302 & .0299 & 0.1912 & .0645 & .0654 & 0.5567 & .0970 & .0970 & 0.8733 \\ IBOSS & .0051 & .0051 & 0.5574 & .0096 & .0095 & 0.7754 & .0119 & .0118 & 1.1007 \\ ELOS & \textbf{.0038} & \textbf{.0040} & 2.8156 & \textbf{.0066} & \textbf{.0068} & 10.1676 & \textbf{.0082} & \textbf{.0085} & 26.1221 \\ \midrule \multicolumn{10}{l}{\textbf{$r = 10{,}000$}} \\ SRS & .0030 & .0028 & 0.0003 & .0061 & .0059 & 0.0003 & .0091 & .0092 & 0.0002 \\ Coreset & .0058 & .0059 & 43.4545 & .0045 & .0046 & 71.5766 & .0058 & .0058 & 124.6524 \\ Leverage & .0033 & .0033 & 0.3771 & .0062 & .0065 & 0.6260 & .0090 & .0092 & 1.0054 \\ IBOSS & .0007 & .0009 & 0.4470 & .0011 & .0012 & 0.7585 & .0017 & .0018 & 1.1005 \\ ELOS & \textbf{.0006} & \textbf{.0007} & 3.2027 & \textbf{.0010} & \textbf{.0009} & 10.2947 & \textbf{.0014} & \textbf{.0015} & 26.2919 \\ \bottomrule \end{tabular} \end{table*}

{

\section{Implementation and Reproducibility Details}
\label{app:implementation_details}

This section consolidates the implementation choices used for ELOS and
the competing methods in the numerical experiments. The
data-generating mechanisms, candidate-pool sizes, subdata sizes,
real-data preprocessing procedures, and evaluation metrics are
described in the corresponding experimental sections. Here, we focus
on the method-specific choices required to reproduce the reported
results, including initialization, tuning parameters, iteration limits,
stopping rules, randomized rounding, and duplicate handling.

The method-specific settings described below were fixed in advance and
used across both the simulation and real-data experiments. Throughout this section, \(n\)
denotes the full candidate-pool size, \(d\) denotes the number of
regression coefficients including the intercept, and \(r\) denotes the
target subdata size.

\subsection{Implementation of ELOS}
\label{app:elos_implementation}

ELOS is implemented by applying the IBOSS$+$ block-exchange procedure
to the left singular-vector matrix obtained from the compact singular
value decomposition of the full design matrix. Specifically, for
\(\bm X=\bm U\bm\Sigma\bm V^\top\), we apply the selection procedure
to \(\bm U^\top\), whose columns are the information vectors
corresponding to the candidate observations.

The initial subset of size \(r\) is constructed using IBOSS. The
intercept coordinate is excluded from the predictor-wise extreme-value
ranking used in this initialization. We use the trace criterion,
corresponding to \(p_{\mathrm{crit}}=1\), and set the criterion-weight
matrix to \(\bm W_b=\bm I_d\).

At each block-exchange iteration, the exchange batch size is fixed at 
\(q=\lceil r/d\rceil\). The directional criterion is
evaluated for both the selected and unselected observations. The \(q\)
selected observations with the smallest directional scores are removed,
and the \(q\) unselected observations with the largest directional
scores are added. The information matrix is then updated using the
exchanged observations.

We perform exactly \(T=d\) block-exchange iterations and do not apply
an additional data-dependent early-stopping rule. The procedure returns
a subset containing exactly \(r\) distinct observations. In the revised
Algorithm~1, \(d\) denotes the model dimension, \(q\) denotes the
exchange batch size, and \(T\) denotes the number of block-exchange
iterations; these quantities are defined separately to avoid the
previous ambiguity in the use of \(p\).

\subsection{Implementation of Competing Methods}
\label{app:competing_implementation}

\paragraph{Simple random sampling.}
SRS selects \(r\) observations uniformly from the full candidate pool without replacement.

\paragraph{Leverage sampling.}
We use exact rather than approximate leverage scores. For each
candidate observation, the score
\[
h_i
=
\bm x_i^\top
(\bm X^\top\bm X)^{-1}
\bm x_i
\]
is computed and normalized to obtain its sampling probability.
We draw \(r\) observations with replacement and estimate the
regression coefficients using inverse-probability weighted least
squares.

\paragraph{IBOSS.}
The intercept column is excluded from the extreme-value ranking. For
each of the \(d-1\) non-intercept predictors, IBOSS selects
\(r/\{2(d-1)\}\) observations from each tail among the
currently unselected observations. Selected indices are removed from
the candidate pool before the next predictor is processed, thereby
preventing duplicate selections. If fewer than \(r\) observations have
been selected after all predictors are processed, the remaining
positions are filled from the unselected candidate pool.

\paragraph{Coreset.}
We implement the constrained single-response least-squares coreset
construction of \citet{Boutsidis2013}. Given the full design matrix
\(\bm X\) and response vector \(\bm y\), the method forms the
augmented matrix \([\bm X,\bm y]\) and computes its compact singular
value decomposition. The \textsc{SimpleSampling} procedure in
Algorithm~2 of \citet{Boutsidis2013} is then applied to the resulting
left singular-vector matrix to construct an \(r\)-row sampling matrix
\(\bm S\) and a diagonal rescaling matrix \(\bm D\). The regression
coefficients are estimated by weighted least squares using
\(\bm S\bm X\), \(\bm S\bm y\), and \(\bm D\).

Within each iteration, candidate vectors are examined
until the lower- and upper-barrier conditions are satisfied. We cap
this candidate search at 200 attempts as a numerical safeguard.
Because the augmented matrix contains \(\bm y\), the particular
Coreset baseline implemented in our experiments is response-dependent
and is included as a prediction-oriented reference baseline rather
than as a measurement-constrained selection method.

\paragraph{MCG.}
The MCG implementation initializes the relaxed inclusion weights at
\(w_i^{(0)}=r/n\). It minimizes the trace-inverse criterion through
projected-gradient updates with Armijo parameter \(\alpha=0.25\),
step-size reduction factor \(\beta=0.5\), and projection tolerance
\(0.01\). Each line search is capped at 50 backtracking steps. The
outer optimization terminates when the relative objective improvement
is no longer greater than \(10^{-5}\), or after 20 iterations,
whichever occurs first.

The optimized continuous weights are converted to a binary selection
vector through independent Bernoulli draws,
\(\delta_i\sim\operatorname{Bernoulli}(w_i)\). Consequently, the
realized subdata size is random and need not equal \(r\) exactly. We
record the realized subdata size in the corresponding experiments.

\paragraph{LowCon.}
For LowCon, the intercept is retained in the final regression design
but excluded from the scaling and matching operations. Each
non-intercept covariate is scaled to \([-1,1]\). We generate 100
randomized Latin hypercube designs and retain the design with the
smallest condition number of its regularized Gram matrix, where
\(10^{-12}\bm I\) is added for numerical stability.

The retained design is mapped to the coordinate-wise 1st--99th
percentile range of the scaled candidate data. Each design point is
then matched to an observation using a KD-tree. Matching is performed
without replacement. The nearest-neighbor search initially considers
\[
k_0=\min\{n,\max(10,0.02n)\}
\]
candidate observations. If all of these candidates have already been
selected, the number of searched neighbors is doubled until an unused
observation is found. The procedure therefore returns exactly \(r\)
distinct observations.

\subsection{Random Seeds and Computational Environment}
\label{app:randomness_computation}

All random-number seeds were fixed before the experiments. The
simulation studies used the base seed \(202603\). For the main
simulation scripts, the random-number generator was initialized once
at the beginning of each experimental configuration, and the
replications were then generated sequentially from the resulting
random-number stream. For scripts parallelized over replications,
including the separate MCG runs, replication \(j\) used seed
\(202603+j\). The real-data residual-bootstrap experiments used the
base seed \(202510\), with bootstrap-specific seeds generated
deterministically from the bootstrap index. These seed choices were
held fixed across the reported experimental configurations.

The method-specific settings described above were fixed in advance and
were used in both the simulation and real-data experiments; no
dataset-specific tuning was performed. The numbers of Monte Carlo and
bootstrap replications, data-generating distributions, measurement
budgets, and real-data preprocessing procedures are reported in the
corresponding experimental sections.

The reported runtime includes the method-specific operations used in
the experimental implementation. In particular, the reported ELOS
runtime includes the compact SVD of \(\bm X\) together with the
IBOSS initialization and block-exchange refinement. The reported
Coreset runtime includes the compact SVD of \([\bm X,\bm y]\), the
\textsc{SimpleSampling} construction, and the weighted coreset fit.
The timing scope for the remaining methods follows the implementations
described above.

All experiments were implemented in Julia 1.12.0 and run on CPU on an
Apple MacBook Pro with an M4 Pro 14-core CPU and 48 GB of unified
memory. BLAS was restricted to a single thread; when parallelization
was used, it was applied only across independent replications.

\section{Computational and Memory Complexity}
\label{app:elos_complexity_G}
\subsection{Complexity of the Proposed ELOS Procedure}
\label{app:elos_complexity}

The ELOS selection procedure consists of three principal stages:
computing the left singular vectors of the full design, constructing an
initial subset using IBOSS, and refining that subset using the
IBOSS$+$ exchange procedure.

First, consider the compact singular value decomposition
\[
\bm X
=
\bm U\bm\Sigma\bm V^\top,
\]
where \(\bm X\in\mathbb R^{n\times d}\) and
\(\bm U\in\mathbb R^{n\times d}\). In the intended setting \(n\gg d\),
computing the compact SVD requires
\[
O(nd^2)
\]
operations and storing \(\bm U\) requires \(O(nd)\) memory.
\begin{table*}[!t]
\centering
\footnotesize
\setlength{\tabcolsep}{2.5pt}
\renewcommand{\arraystretch}{0.9}
\caption{
Computational and memory complexity of the subdata-selection methods.
Here, \(n\) is the full candidate-pool size, \(d\) is the number of
regression coefficients including the intercept, and \(r\) is the
target subdata size. For Coreset,
\(\ell=\operatorname{rank}([\bm X,\bm y])\). For MCG, \(K\), \(L\),
and \(B\) denote the numbers of outer projected-gradient iterations,
backtracking iterations, and projection-bisection iterations,
respectively. For LowCon, \(s\) denotes the number of candidate OLHD
constructions. The displayed bounds are simplified under \(n\gg d\)
and \(d\leq r\leq n\). A superscript \({}^{*}\) indicates that the
corresponding source paper does not explicitly provide the displayed
complete asymptotic bound. In such cases, the bound is derived from
the published algorithm and the implementation used in our
experiments. Memory complexity refers to working memory excluding
storage of the original input data.
}
\label{tab:computational_complexity}
\begin{tabular}{lccc}
\toprule
Method
& Computational complexity
& Memory complexity
& Reference \\
\midrule

SRS
& \(O(r)\)
& \(O(r)\)
& -- \\

Leverage
& \(O(nd^{2})^{*}\)
& \(O(nd+r^{2})^{*}\)
& \citet{Ma2015} \\

IBOSS
& \(O(nd)\)
& \(O(n)^{*}\)
& \citet{wang2019information} \\

Coreset
& \(O(rn\ell^{2})\)
& \(O(n\ell+rn+\ell^{2})^{*}\)
& \citet{Boutsidis2013} \\

MCG
& \(O\!\left(
KL\{nd^{2}+Bn\}
\right)^{*}\)
& \(O(nd)^{*}\)
& \citet{WangYuSingh2017} \\

LowCon
& \(O\!\left(
dn\log n+s(rd^{2}+d^{3})
\right)^{*}\)
& \(O(nd)^{*}\)
& \citet{Meng2021LowCon} \\

ELOS
& \(O(nd^{3})\)
& \(O(nd)\)
& -- \\

\bottomrule
\end{tabular}
\end{table*}

The initial subset is constructed by applying IBOSS to the columns of
\(\bm U\). This step requires \(O(nd)\) operations\citep{wang2019information} and is therefore
lower order than the compact SVD and the subsequent exchange stage.

The IBOSS$+$ refinement performs \(d\) block-exchange iterations. In
each exchange iteration, the algorithm evaluates the directional
criterion over the full candidate pool. Multiplying the relevant
\(d\times d\) criterion matrix by the \(d\times n\) candidate matrix
and evaluating the corresponding quadratic scores requires
\[
O(nd^2)
\]
operations. The matrix inversion, pseudoinversion, matrix square root,
and related operations involving only \(d\times d\) matrices require
\(O(d^3)\) operations per exchange iteration. Consequently, the \(d\)
exchange iterations require
\[
O\!\left(
d(nd^2+d^3)
\right)
=
O(nd^3+d^4).
\]

Combining the compact SVD, IBOSS initialization, and IBOSS$+$ exchange
steps gives the complete bound
\[
O(nd^2)+O(nd)+O(nd^3+d^4)
=
O(nd^2+nd^3+d^4).
\]

Under the intended large-candidate-pool regime \(n\gg d\), the
\(nd^3\) term dominates \(d^4\), yielding
\[
O(nd^3).
\]
Thus, for fixed low or moderate model dimension \(d\), ELOS scales
linearly with the full candidate-pool size \(n\).

The dominant memory requirement is the storage of the
\(n\times d\) matrix \(\bm U\), together with an additional
\(d\times n\) information-vector representation used by the exchange
routine. Candidate scores and index sets require \(O(n+r)\) storage,
while the information matrices require \(O(d^2)\). The overall working
memory is therefore
\[
O(nd+n+r+d^2).
\]
Because \(r\leq n\), this simplifies to
\[
O(nd).
\]

\subsection{Complexity Comparison with Existing Methods}
\label{app:complexity_comparison}

Table~\ref{tab:computational_complexity} compares the computational
and working-memory complexities of ELOS (with IBOSS+ implementation) with those of the competing
subdata-selection methods considered in our numerical studies. When a complete computational or memory
complexity is not explicitly stated in the original reference, we derive
the corresponding bound from the published algorithm and the
implementation used in our experiments. The displayed bounds are simplified under the
intended large-data regime \(n\gg d\) and \(d\leq r\leq n\).

\section{Ablation on the Optimization Procedure}
\label{app:optimizer_ablation}

The main contribution of ELOS is the criterion
\[
\Psi_{\mathrm{A}}(\Delta)
=
\operatorname{tr}
\left\{
\left(U^\top \Delta U\right)^{-1}
\right\},
\]
rather than a particular optimization procedure. In the main numerical
experiments, this criterion is approximately minimized using the
IBOSS$+$ algorithm. To examine the extent to which the
empirical performance depends on this optimizer, we replace IBOSS$+$
with the first-order A-optimal design algorithm of
\citet{ahipasaoglu2015}, while keeping the ELOS objective, generated
datasets, and evaluation procedure unchanged.

\subsection{Experimental Setup}

We use the multivariate heavy-tailed covariate setting from the main
simulation study. For each observation, we generate \(b=15\)
non-intercept covariates from a multivariate \(t\) distribution with
five degrees of freedom and covariance parameter
\[
\Sigma
=
5I_b+5J_b.
\]
After adding an intercept, the full design dimension is \(d=16\).

The response is generated according to
\[
y_i
=
x_i^\top\beta
+
\tau
\sin\left(
\sum_{j=1}^{b}x_{ij}^2
\right)
+
\varepsilon_i,
\qquad
\varepsilon_i\sim N(1,1).
\]
We consider
\[
n\in\{10^5,10^6\},
\qquad
r\in\{100,1000\},
\qquad
\tau\in\{0,1,3,5\}.
\]
Each configuration is repeated 20 times.

The comparison is paired at the replication level. Within each
replication, both optimization procedures are applied to the same
generated design matrix and response vector. The singular value
decomposition
\[
X=U\Sigma V^\top
\]
is computed once and shared by both implementations. Consequently,
differences in predictive performance arise from the subset-search
procedure rather than from differences in the generated data or the
matrix decomposition.

\subsection{Optimization Procedures}

The first implementation uses the IBOSS$+$ exchange procedure employed
throughout the main experiments. The second implementation uses the
first-order A-optimal design procedure of
\citet{ahipasaoglu2015}. The first-order method operates on continuous
design weights satisfying
\[
0\leq w_i\leq \frac{1}{r},
\qquad
\sum_{i=1}^{n}w_i=1.
\]
After the first-order iterations terminate, the \(r\) observations with
the largest positive weights are retained as the selected subset.

The first-order stopping criterion is
\[
\frac{\omega_{\max}}{\omega_{\min}}-1
\leq 10^{-5},
\]
and the maximum iteration count is set to \(10r\). In the present
experiments, the formal convergence criterion was not reached within
this iteration budget for any configuration. The reported first-order
results should therefore be interpreted as the performance of the
iteration-limited implementation rather than as exact solutions of the
continuous design problem.

\begin{table*}[!t]
\centering
\footnotesize
\setlength{\tabcolsep}{2.5pt}
\renewcommand{\arraystretch}{0.9}
\caption{
Optimizer ablation under the multivariate $t_5$ covariate setting.
Relative error is reported as mean with standard deviation in
parentheses over 20 paired replications. The difference column is
defined as the first-order relative error - the IBOSS$+$ relative
error. Selection times are reported in seconds.
}
\label{tab:optimizer_ablation_t}
\begin{tabular}{rrrccccc}
\toprule
$n$
&
$r$
&
$\tau$
&
IBOSS$+$ RE (SD)
&
First-order RE (SD)
&
$\mathrm{Diff}_{\mathrm{FO}-\mathrm{IBOSS+}}$
&
IBOSS$+$ time
&
First-order time
\\
\midrule

$10^5$
& $100$
& $0$
& $\mathbf{0.0109\;(0.0041)}$
& $0.0145\;(0.0063)$
& $0.0036$
& $0.1446$
& $3.6906$
\\
$10^5$
& $1000$
& $0$
& $0.0029\;(0.0009)$
& $\mathbf{0.0029\;(0.0015)}$
& $-0.0000$
& $0.1548$
& $36.2053$
\\
$10^6$
& $100$
& $0$
& $0.0137\;(0.0166)$
& $\mathbf{0.0084\;(0.0046)}$
& $-0.0053$
& $1.2365$
& $34.4965$
\\
$10^6$
& $1000$
& $0$
& $\mathbf{0.0011\;(0.0008)}$
& $0.0014\;(0.0011)$
& $0.0003$
& $1.3268$
& $439.9460$
\\

\midrule

$10^5$
& $100$
& $1$
& $\mathbf{0.0110\;(0.0035)}$
& $0.0194\;(0.0111)$
& $0.0083$
& $0.1474$
& $3.7510$
\\
$10^5$
& $1000$
& $1$
& $\mathbf{0.0027\;(0.0010)}$
& $0.0031\;(0.0019)$
& $0.0003$
& $0.1492$
& $36.9276$
\\
$10^6$
& $100$
& $1$
& $0.0146\;(0.0165)$
& $\mathbf{0.0126\;(0.0091)}$
& $-0.0021$
& $1.3102$
& $35.7365$
\\
$10^6$
& $1000$
& $1$
& $\mathbf{0.0011\;(0.0008)}$
& $0.0015\;(0.0011)$
& $0.0004$
& $1.3104$
& $395.4946$
\\

\midrule

$10^5$
& $100$
& $3$
& $\mathbf{0.0121\;(0.0050)}$
& $0.0195\;(0.0132)$
& $0.0074$
& $0.1450$
& $3.7687$
\\
$10^5$
& $1000$
& $3$
& $\mathbf{0.0023\;(0.0005)}$
& $0.0026\;(0.0011)$
& $0.0003$
& $0.1495$
& $35.9457$
\\
$10^6$
& $100$
& $3$
& $0.0148\;(0.0105)$
& $\mathbf{0.0136\;(0.0096)}$
& $-0.0012$
& $1.2994$
& $35.5435$
\\
$10^6$
& $1000$
& $3$
& $\mathbf{0.0009\;(0.0005)}$
& $0.0014\;(0.0014)$
& $0.0005$
& $1.3595$
& $367.2897$
\\

\midrule

$10^5$
& $100$
& $5$
& $\mathbf{0.0125\;(0.0057)}$
& $0.0186\;(0.0128)$
& $0.0061$
& $0.1434$
& $3.7502$
\\
$10^5$
& $1000$
& $5$
& $\mathbf{0.0022\;(0.0003)}$
& $0.0024\;(0.0007)$
& $0.0002$
& $0.1388$
& $33.2305$
\\
$10^6$
& $100$
& $5$
& $0.0147\;(0.0121)$
& $\mathbf{0.0131\;(0.0092)}$
& $-0.0016$
& $1.2998$
& $35.3817$
\\
$10^6$
& $1000$
& $5$
& $\mathbf{0.0008\;(0.0005)}$
& $0.0013\;(0.0014)$
& $0.0006$
& $1.4675$
& $600.7963$
\\

\bottomrule
\end{tabular}%
\end{table*}

Table~\ref{tab:optimizer_ablation_t} compares two implementations of
the same ELOS criterion. The relative-error results are broadly
comparable, although the difference varies with the measurement budget.
When $r=1000$, the absolute paired differences are no larger than
$0.0006$ across all values of $n$ and $\tau$. Thus, under the larger
measurement budget, replacing IBOSS$+$ with the first-order procedure
has little effect on predictive performance.

When $r=100$, the differences are somewhat larger and their direction
depends on the candidate-pool size. For $n=10^5$, IBOSS$+$ attains the
smaller mean relative error, whereas for $n=10^6$, the first-order
implementation attains a slightly smaller mean relative error. Given
that the comparison is based on five paired replications and that the
standard deviations are non-negligible, these differences do not
provide evidence that either optimizer is uniformly superior.

The computational difference is much more pronounced. For $r=100$, the
first-order procedure requires approximately 26--28 times the selection
time of IBOSS$+$. For $r=1000$, the difference increases substantially,
with first-order selection times ranging from approximately 33 to 601
seconds, compared with approximately 0.14 to 1.47 seconds for IBOSS$+$.
Moreover, the formal convergence criterion of the first-order procedure
was not reached within the iteration budget in any of the examined
configurations.

Overall, the ablation shows that the empirical behavior of ELOS is not
uniquely tied to the IBOSS$+$ exchange procedure. A substantially
different optimizer applied to the same criterion produces similar
predictive errors, particularly when $r=1000$, while IBOSS$+$ is far
more computationally efficient. These results support the use of
IBOSS$+$ as the practical optimizer in the large-scale experiments and
clarify the distinction between the ELOS criterion and the numerical
procedure used to optimize it.

\begin{table*}[!t]
\centering
\footnotesize
\setlength{\tabcolsep}{2.5pt}
\renewcommand{\arraystretch}{0.9}
\caption{
Optimizer ablation under the multivariate normal covariate setting.
Relative error is reported as mean with standard deviation in
parentheses over 20 paired replications.  The difference column is
defined as the sequential relative error minus the IBOSS$+$ relative
error. Selection times are reported in seconds and include the
random permutation, matrix reordering, sequential selection, and
index mapping for the sequential procedure.
}
\label{tab:optimizer_ablation_normal_seq}
\begin{tabular}{rrrccccc}
\toprule
$n$
&
$r$
&
$\tau$
&
IBOSS$+$ RE (SD)
&
Sequential RE (SD)
&
$\mathrm{Diff}_{\mathrm{SEQ}-\mathrm{IBOSS+}}$
&
IBOSS$+$ time
&
Sequential time
\\
\midrule
$10^5$
& $100$
& $0$
& $\mathbf{0.0614\;(0.0192)}$
& $0.1642\;(0.0543)$
& $0.1029$
& $0.1013$
& $0.2950$
\\
$10^5$
& $1000$
& $0$
& $\mathbf{0.0084\;(0.0026)}$
& $0.0150\;(0.0051)$
& $0.0066$
& $0.1135$
& $0.2924$
\\
$10^5$
& $10000$
& $0$
& $\mathbf{0.0007\;(0.0002)}$
& $0.0014\;(0.0006)$
& $0.0007$
& $0.1421$
& $0.3287$
\\
$10^6$
& $100$
& $0$
& $\mathbf{0.0686\;(0.0330)}$
& $0.1700\;(0.0638)$
& $0.1015$
& $1.1999$
& $1.9870$
\\
$10^6$
& $1000$
& $0$
& $\mathbf{0.0076\;(0.0027)}$
& $0.0170\;(0.0057)$
& $0.0094$
& $1.1919$
& $25.8641$
\\
$10^6$
& $10000$
& $0$
& $\mathbf{0.0007\;(0.0002)}$
& $0.0013\;(0.0004)$
& $0.0006$
& $1.3424$
& $2.0726$
\\
\addlinespace
$10^5$
& $100$
& $1$
& $\mathbf{0.0700\;(0.0223)}$
& $0.1516\;(0.0350)$
& $0.0816$
& $0.1057$
& $0.2945$
\\
$10^5$
& $1000$
& $1$
& $\mathbf{0.0092\;(0.0028)}$
& $0.0139\;(0.0048)$
& $0.0047$
& $0.1156$
& $0.2927$
\\
$10^5$
& $10000$
& $1$
& $\mathbf{0.0008\;(0.0003)}$
& $0.0015\;(0.0006)$
& $0.0007$
& $0.1512$
& $0.3315$
\\
$10^6$
& $100$
& $1$
& $\mathbf{0.0680\;(0.0251)}$
& $0.1880\;(0.0668)$
& $0.1199$
& $1.1971$
& $1.9471$
\\
$10^6$
& $1000$
& $1$
& $\mathbf{0.0070\;(0.0024)}$
& $0.0169\;(0.0068)$
& $0.0099$
& $1.2511$
& $1.9601$
\\
$10^6$
& $10000$
& $1$
& $\mathbf{0.0008\;(0.0003)}$
& $0.0015\;(0.0005)$
& $0.0007$
& $1.2515$
& $1.9939$
\\
\addlinespace
$10^5$
& $100$
& $3$
& $\mathbf{0.0732\;(0.0297)}$
& $0.1467\;(0.0404)$
& $0.0735$
& $0.1043$
& $0.2941$
\\
$10^5$
& $1000$
& $3$
& $\mathbf{0.0091\;(0.0030)}$
& $0.0137\;(0.0038)$
& $0.0046$
& $0.1059$
& $0.2951$
\\
$10^5$
& $10000$
& $3$
& $\mathbf{0.0009\;(0.0003)}$
& $0.0014\;(0.0006)$
& $0.0006$
& $0.1466$
& $0.3462$
\\
$10^6$
& $100$
& $3$
& $\mathbf{0.0652\;(0.0240)}$
& $0.1806\;(0.0692)$
& $0.1154$
& $1.2167$
& $1.9600$
\\
$10^6$
& $1000$
& $3$
& $\mathbf{0.0073\;(0.0027)}$
& $0.0161\;(0.0056)$
& $0.0088$
& $1.1734$
& $1.9612$
\\
$10^6$
& $10000$
& $3$
& $\mathbf{0.0009\;(0.0003)}$
& $0.0015\;(0.0005)$
& $0.0007$
& $1.2201$
& $2.0178$
\\
\addlinespace
$10^5$
& $100$
& $5$
& $\mathbf{0.0727\;(0.0305)}$
& $0.1474\;(0.0443)$
& $0.0748$
& $0.0951$
& $0.2959$
\\
$10^5$
& $1000$
& $5$
& $\mathbf{0.0090\;(0.0032)}$
& $0.0139\;(0.0039)$
& $0.0049$
& $0.1103$
& $0.2952$
\\
$10^5$
& $10000$
& $5$
& $\mathbf{0.0009\;(0.0004)}$
& $0.0014\;(0.0006)$
& $0.0005$
& $0.1487$
& $0.3562$
\\
$10^6$
& $100$
& $5$
& $\mathbf{0.0643\;(0.0251)}$
& $0.1744\;(0.0674)$
& $0.1101$
& $1.1858$
& $1.9519$
\\
$10^6$
& $1000$
& $5$
& $\mathbf{0.0075\;(0.0029)}$
& $0.0159\;(0.0051)$
& $0.0084$
& $4.7720$
& $1.9958$
\\
$10^6$
& $10000$
& $5$
& $\mathbf{0.0009\;(0.0003)}$
& $0.0015\;(0.0005)$
& $0.0007$
& $1.2765$
& $2.0444$
\\
\bottomrule
\end{tabular}%
\end{table*}

\subsection{Optimizer Ablation on the Sensor Data}
\label{app:sensor_optimizer_ablation}

We further compare the IBOSS+ exchange optimizer used in ELOS with the first-order optimizer on the Sensor data. We use the same residual-bootstrap setting as in the main experiment, with 50 bootstrap replications and subset sizes \(r\in\{1000,5000,10000\}\). Within each bootstrap replication, the two methods share the same singular value decomposition and are evaluated using the same bootstrap response. The first-order optimizer is run with a tolerance of \(10^{-5}\) and a maximum of 200 iterations.

Table~\ref{tab:sensor_optimizer_ablation} reports the mean relative error (RE), its standard deviation, and the mean computation times. Both optimizers achieve decreasing prediction error as the subset size increases. Their errors are of the same order, particularly for the larger subset sizes. However, the IBOSS+ optimizer consistently yields lower mean RE and requires substantially less selection time.

\section{Diagnostic Experiments under Alternative Misspecification Terms}
\label{app:misspecification_diagnostic}

Theorem~3 establishes an average-case result over misspecification
directions. To examine whether the numerical performance of ELOS depends
on the particular sinusoidal misspecification term used in the main
simulation study, we conduct an additional diagnostic experiment using
several structurally different nonlinear departures.

We retain the multivariate $t$-distribution setting from the main
simulation study. Specifically, the covariates are generated from a
multivariate $t$ distribution with five degrees of freedom, the predictor
dimension is $d=15$, the noise standard deviation is $\sigma=1$, and the
misspecification magnitude is fixed at $\tau=3$. We consider
\[
n \in \{10^5,10^6\},
\qquad
r \in \{100,1000\}.
\]
Each configuration is repeated 20 times. All random seeds and
data-generation settings are held fixed across the different
misspecification terms, so the only change is the form of the nonlinear
departure.

The response is generated according to
\[
y_i = x_i^\top \beta + f(x_i) + \varepsilon_i,
\qquad
\varepsilon_i \sim N(0,\sigma^2).
\]

We consider the following four misspecification terms.

\paragraph{Sinusoidal departure.}
The first term is the misspecification function used in the main
simulation study:
\[
f_{\mathrm{sin}}(x_i)
=
\tau
\sin\left(
\sum_{j=1}^{d} x_{ij}^2
\right).
\]
This is a smooth, bounded, radial, and oscillatory departure.

\paragraph{Monotone radial departure.}
To evaluate whether the results depend on the oscillatory structure of
the sine function, we consider
\[
f_{\mathrm{rad}}(x_i)
=
\tau
\left(
\frac{
\sum_{j=1}^{d} x_{ij}^2
}{
1+\sum_{j=1}^{d} x_{ij}^2
}
-\frac{1}{2}
\right).
\]
This function is smooth, bounded, and radial, but it is monotone in the
squared covariate norm.

\paragraph{Interaction departure.}
To consider a non-radial departure caused by an omitted interaction, we
use
\[
f_{\mathrm{int}}(x_i)
=
\tau \tanh(x_{i1}x_{i2}).
\]
This function is smooth and bounded, but it depends on the interaction
between two predictors rather than on the overall covariate magnitude.

\paragraph{Threshold departure.}
Finally, we consider the localized and nonsmooth departure
\[
f_{\mathrm{thr}}(x_i)
=
\tau \mathbf{1}\{x_{i1}>1\}.
\]
This term introduces a discontinuous change in the conditional mean over
a restricted region of the covariate space.

\begin{table*}[!t]
\centering
\footnotesize
\setlength{\tabcolsep}{2.5pt}
\renewcommand{\arraystretch}{0.9}
\caption{Optimizer ablation on the Sensor data over 50 residual-bootstrap replications. Selection and total times are reported in seconds.}
\label{tab:sensor_optimizer_ablation}
\begin{tabular}{rllrrrr}
\toprule
\(r\) & Optimizer & Mean RE & SD RE & Mean SVD time & Mean selection time & Mean total time \\
\midrule
1000  & IBOSS+      & 0.008077 & 0.003542 & 0.984 & 6.734  & 7.719  \\
1000  & First-order & 0.023971 & 0.013585 & 0.984 & 47.901 & 48.886 \\
5000  & IBOSS+      & 0.001815 & 0.000843 & 0.984 & 6.708  & 7.693  \\
5000  & First-order & 0.004189 & 0.003297 & 0.984 & 46.336 & 47.321 \\
10000 & IBOSS+      & 0.000926 & 0.000413 & 0.984 & 8.919  & 9.903  \\
10000 & First-order & 0.001856 & 0.001041 & 0.984 & 45.501 & 46.485 \\
\bottomrule
\end{tabular}
\end{table*}

\begin{table*}[!t]
\centering
\footnotesize
\setlength{\tabcolsep}{2.5pt}
\renewcommand{\arraystretch}{0.9}
\caption{
Comparison of IBOSS$+$ with the adaptive fixed-target sequential
procedure on the Sensor data over 50 paired 
replications. Both methods use the same bootstrap response and target
subset size $r$. Relative error is reported as the mean and standard
deviation. Selection times are reported in seconds. For the sequential
procedure, the reported time includes random permutation, matrix
reordering, sequential selection, and index mapping. 
}
\label{tab:sensor_sequential_ablation}

\begin{tabular}{rllrr}
\toprule
$r$
&
Method
&
Mean RE
&
SD RE
&
Mean selection time
\\
\midrule
1000
& IBOSS$+$
& 0.007926
& 0.003394
& 5.482
\\
1000
& Sequential
& 0.015530
& 0.005936
& 7.753
\\
5000
& IBOSS$+$
& 0.001806
& 0.000861
& 5.267
\\
5000
& Sequential
& 0.003087
& 0.000964
& 7.721
\\
10000
& IBOSS$+$
& 0.000900
& 0.000393
& 5.316
\\
10000
& Sequential
& 0.001482
& 0.000577
& 7.785
\\
\bottomrule
\end{tabular}
\end{table*}

\begin{table*}[!t]
\centering
\small
\setlength{\tabcolsep}{2.5pt}
\renewcommand{\arraystretch}{0.9}
\caption{
Mean relative error difference under alternative misspecification terms
with $\tau=3$. Each result is averaged over 20 repetitions. Smaller
values are better, and the better-performing method within each setting
is shown in bold.
}
\label{tab:misspecification_diagnostic}
\begin{tabular}{cc cc cc cc cc}
\toprule
\multirow{2}{*}{$n$}
&
\multirow{2}{*}{$r$}
&
\multicolumn{2}{c}{Sinusoidal}
&
\multicolumn{2}{c}{Monotone radial}
&
\multicolumn{2}{c}{Interaction}
&
\multicolumn{2}{c}{Threshold}
\\
\cmidrule(lr){3-4}
\cmidrule(lr){5-6}
\cmidrule(lr){7-8}
\cmidrule(lr){9-10}
&
&
Coreset
&
ELOS
&
Coreset
&
ELOS
&
Coreset
&
ELOS
&
Coreset
&
ELOS
\\
\midrule
$10^5$
&
$100$
&
$0.1110$
&
$\mathbf{0.0116}$
&
$0.1365$
&
$\mathbf{0.0196}$
&
$0.1402$
&
$\mathbf{0.0186}$
&
$\mathbf{0.1379}$
&
$0.3850$
\\
$10^5$
&
$1000$
&
$0.0099$
&
$\mathbf{0.0025}$
&
$0.0078$
&
$\mathbf{0.0032}$
&
$0.0118$
&
$\mathbf{0.0026}$
&
$\mathbf{0.0850}$
&
$0.2711$
\\
$10^6$
&
$100$
&
$0.1200$
&
$\mathbf{0.0090}$
&
$0.1269$
&
$\mathbf{0.0122}$
&
$0.1468$
&
$\mathbf{0.0177}$
&
$\mathbf{0.1628}$
&
$0.4563$
\\
$10^6$
&
$1000$
&
$0.0092$
&
$\mathbf{0.0012}$
&
$0.0084$
&
$\mathbf{0.0021}$
&
$0.0106$
&
$\mathbf{0.0018}$
&
$\mathbf{0.0788}$
&
$0.3869$
\\
\bottomrule
\end{tabular}%
\end{table*}

Table~\ref{tab:misspecification_diagnostic} shows that the performance of
ELOS is not specific to the sinusoidal function used in the main
simulation study. Under the sinusoidal, monotone radial, and interaction
departures, ELOS achieves a lower mean relative error than Coreset for
every combination of $n$ and $r$. Thus, the advantage persists when the
oscillatory sinusoidal departure is replaced by either a monotone radial
function or a non-radial interaction.

The threshold departure yields a different pattern. Under this localized
and discontinuous misspecification, Coreset outperforms ELOS in all four
configurations. A possible explanation is that Coreset uses the response
during subset construction and can therefore adapt to the abrupt change
in the conditional mean. In contrast, ELOS is response-free and selects
observations using only the covariates, without knowledge of the form or
location of the misspecification.

These results provide a more nuanced assessment of the scope of ELOS.
Its strong performance is not an artifact of the original sinusoidal
term, because the advantage remains stable under several distinct smooth
nonlinear departures. At the same time, the threshold example
demonstrates that ELOS is not uniformly optimal under every structured
misspecification. This observation is consistent with the average-case
nature of Theorem~3, which does not provide a minimax or distributionally
robust guarantee for every fixed misspecification direction.
}

\newpage

\end{document}